\documentclass[journal]{IEEEtran}

\IEEEoverridecommandlockouts
\usepackage{cite}
\usepackage{amsmath,amssymb,amsfonts}
\usepackage{algorithmic}
\usepackage{graphicx}
\usepackage{textcomp}
\usepackage{xcolor}
\usepackage{amsthm}
\usepackage{overpic}
\usepackage{bm}
\usepackage{float}
\usepackage{amsthm}
\usepackage{mathtools}
\newcommand{\condSum}[3]{\overset{#3}{\underset{\underset{#2}{#1}}{\sum}}}
\DeclareMathOperator{\sinc}{sinc}
\usepackage{subcaption}
\usepackage{comment}

\theoremstyle{plain}
\newtheorem{theorem}{Theorem}

\newtheorem{lemma}{Lemma}
\newtheorem{remark}{Remark}

\def\Htran{\mbox{\tiny $\mathrm{H}$}}
\def\Ttran{\mbox{\tiny $\mathrm{T}$}}
\def\CN{\mathcal{N}_{\mathbb{C}}} 
\newcommand{\vect}[1]{\mathbf{#1}}

\RenewDocumentCommand{\splitdfrac}{smm}{%
  \genfrac{}{}{0pt}{0}
    {\mathstrut#2\IfBooleanF{#1}{\quad\hfill}}
    {\IfBooleanF{#1}{\hfill\quad}\mathstrut#3}
}

\begin{document}

\title{Optimizing Reconfigurable Intelligent Surfaces for Short Transmissions: How Detailed Configurations can be Afforded?

}

\author{\IEEEauthorblockN{Anders Enqvist, \"Ozlem Tu\u{g}fe Demir, \emph{Member, IEEE}, Cicek Cavdar, \emph{Member, IEEE}, Emil Bj{\"o}rnson, \emph{Fellow, IEEE}}\vspace{-1cm}
\thanks{A. Enqvist, C. Cavdar, and E.~Bj\"ornson are with the Department of Computer Science, KTH Royal Institute of Technology, Kista, Sweden (\{enqv, cavdar, emilbjo\}@kth.se). \newline \indent \"O. T. Demir is with the  Department of Electrical and Electronics Engineering, TOBB University of Economics and Technology, Ankara, Turkey (ozlemtugfedemir@etu.edu.tr). \newline \indent This work was supported by the  FFL18-0277 grant from the Swedish Foundation for Strategic Research.\newline 
\indent The preliminary version \cite{enqvist2022optimizing} of this work was presented at ICC 2022 in Seoul, South Korea.  \newline \indent
© 2023 IEEE. Personal use of this material is permitted. Permission from IEEE must be 
obtained for all other uses, in any current or future media, including 
reprinting/republishing this material for advertising or promotional purposes, creating new 
collective works, for resale or redistribution to servers or lists, or reuse of any copyrighted 
component of this work in other works. DOI:10.1109/TWC.2023.3307605.}
}

\maketitle


\begin{abstract}


This {\color{black}paper examines how to minimize the energy consumption of a user equipment (UE) when transmitting short data payloads. The receiving base station (BS) controls a reconfigurable intelligent surface (RIS), which requires additional pilot signals to be configured, to improve the channel conditions. The challenge is that the pilot signals increase the energy consumption and must be balanced against energy savings during data transmission.
We derive a formula for the energy consumption, including both pilot and data transmission powers and the effects of imperfect channel state information and discrete phase-shifts. 

To shorten the pilot length, we propose dividing the RIS into subarrays of multiple elements using the same reflection coefficient. The pilot power and subarray size are tuned to the payload length to minimize the energy consumption. Analytical results show that there exists a unique energy-minimizing solution. For small payloads and when the direct path loss between the BS and UE is weak compared to the path loss via the RIS, the solution is using subarrays with many elements and low pilot power and vice versa. The optimal percentage of energy spent on pilot signaling is in the order of $10$-$40$\%. 

}

\end{abstract}


\section{Introduction}

Most wireless communication applications generate intermittent traffic\cite{asplund2020advanced}. When a user equipment (UE) should transmit a finite piece of data to a base station (BS), the payload size depends on which modulation and coding scheme is supported, which in turn is determined by the channel conditions. The channel used to be given by nature, but the advent of reconfigurable intelligent surfaces (RISs) makes it partially controllable \cite{8741198,Renzo2020b,Qingqing2021a,bjornson2022SPM}. An RIS consists of an array of sub-wavelength-sized elements that can be configured to reflect incident waves in desired ways \cite{8466374}, for example, as a beam towards the intended receiver. Recent works have described RIS experiments with over a thousand elements \cite{pei2021ris}. The individual wireless links might not be capacity-constrained in {\color{black}the sixth generation (6G) wireless technology}, thanks to enormous bandwidths, but rather limited by energy consumption.

\subsection{Prior Work}
A large RIS can improve the energy efficiency (EE) of the transmission and beats traditional relays when it is sufficiently large \cite{8741198, 8888223,9497709}. These prior works use the capacity-to-power ratio as the EE metric, targeting continuous transmission with full buffers and deterministic channels.
The RIS-aided ergodic capacity of fading channels was studied in \cite{9244106,9333612}, where it was shown that particular numbers of RIS elements maximize the capacity or EE metric (ergodic-capacity-to-power ratio), respectively. A key assumption was that the pilot length equals the number of elements (plus one due to the direct BS-UE channel), thus the channel coherence limits how many RIS elements are useful. However, there exist methods using shorter sequences as well \cite{bjornson2022maximum}. An RIS can help reduce the required uplink (UL) transmit power in a network if the RIS phase-shifts and combining vectors at the BS are optimized \cite{wu2022energy}. In \cite{You2021}, it was shown that a trade-off between spectral efficiency (SE) and EE exists in the UL for RIS-assisted {\color{black} multi-user multiple-input multiple-output (MU-MIMO)} systems and in \cite{Hua2021} how to utilize an RIS to reduce the overall energy consumption in a network of multiple UEs and BSs. In \cite{Ren2022}, a study of RIS enhancements to the reliability of short block length URLLC (ultra-reliable low-latency communication) system was conducted. {\color{black} This paper studies the impact of discrete phase-shifts on energy consumption. An analysis of how the data rate is degraded by discrete phase-shifts can also be found in \cite{zhang2020reconfigurable}. The work \cite{xiong20223d} shows how discrete phase-shifts exhibit very similar behaviors to continuous phase-shifts concerning an analysis of a Doppler power spectral density and autocorrelation function for a 3D (three-dimensional) non-stationary wideband channel model if there are two or more phase-shift bits. Implementations of a 1-bit RIS can be found in \cite{1BITRIS,kayraklik2022indoor}, 2-bit in \cite{dai2020reconfigurable}, and 3-bit in \cite{rains2022high}. 
Several papers have also explored the application of deep reinforcement learning (DRL) to optimization problems involving RIS-assisted systems. In \cite{saglam2022deep}, DRL was used for optimizing the sum-downlink-rate problem in an RIS-aided cellular system under imperfect channel state information (CSI) and hardware impairments. In \cite{9322372}, DRL was used to co-optimize the BS beamforming together with the beamforming of the RIS to minimize the required transmit power for the target user in the downlink. In \cite{zhu2022drl}, DRL was used in an RIS-assisted network to optimize the BS-RIS-UE allocation along with the beamforming.

A closely related technology is active RIS which is a relay-like version with integrated power amplifiers that can actively control the reflected signal amplitude in addition to the phase. There exists a trade-off between the network EE of active and passive RIS because of increased capabilities of boosting the signal-to-noise ratio (SNR) further at the additional cost of increased power consumption in the active hardware. Notable works include \cite{zhi2022active,fotock2023energy,zhang2022active}. Another possibility for improving the EE is through holographic MIMO technology which uses novel metamaterials to create a holographic radio environment to improve the signal \cite{huang2020holographic, 9826717}.

}

\subsection{Motivations}
 A practical obstacle is the additional pilot signaling overhead required to identify the optimal RIS configuration \cite{9311936}. {\color{black} This overhead is particularly problematic when the payload size is small in comparison to the pilot since shorter payloads make the pilot overhead relatively more important. The study of this scenario is a novel contribution of this paper.}
We take a new approach by considering the UL transmission of a finite-sized data payload, as is typically the case in practice \cite{asplund2020advanced}. A short-sized payload, which is an encoded data packet, increases the gap to the Shannon capacity for different code rates \cite{polyanskiy2010channel}, but we still need a certain SNR to achieve a desirably low error probability. We also emphasize the energy consumption that occurs as part of the signaling; how it plays an important role and why it should not be neglected.

The data rate and modulation format are predetermined in this paper, thus the maximum EE is achieved by minimizing the energy consumption at the UE with support by an RIS. This should be done for a given number of RIS elements and under the constraint that the SNR during data transmission is sufficient for the receiver to decode the data. 
The SNR boost that an RIS provides depends on its physical size (aperture gain) and its configuration (beamforming gain).
There are two main ways to limit the pilot signaling.
One is to turn off elements (i.e., turn them into absorbers) so that the number of active elements matches with the number of pilots \cite{9244106,9333612}. This incurs a large SNR loss since both the aperture and beamforming gains are reduced.
We instead consider grouping RIS elements into \emph{subarrays} that share the same reflection coefficient and pilot, similarly to how antenna arrays in 5G group physical antenna elements together as \emph{antenna ports} that use common reference signals. Grouping of RIS elements has been considered in \cite{zheng2019intelligent,you2020channel,demir2022channel} and experimental results can be found in \cite{kayraklik2022indoor}.
The benefit is that only the beamforming gain is reduced. We show that the pilot length can be made equal to the number of subarrays (plus one). 

\subsection{Contributions}
Different from previous works, in this paper we provide new insights through a novel problem of minimizing the \emph{energy consumption} rather than considering the SNR or SE as in the aforementioned papers. This paper aims to answer the following research questions:

\begin{itemize}
    \item Can the energy consumption be reduced by shortening the pilot sequences or reducing their power?
    \item Should the RIS be configured differently depending on the payload length?
    \item How are deployment factors, such as hardware impairments, imperfect CSI, and propagation conditions, impacting the answers to these questions?
\end{itemize}

 {\color{black}  To find the answers, we derive novel analytical closed-form expressions to formulate and solve an optimization problem where the number of subarrays and pilot power are the variables with the objective of minimizing energy consumption. How the solution depends on the payload length, the number of RIS elements, phase-shift quantization, and channel gains is showcased analytically and numerically to provide concrete answers to our research questions. }

This is an extended and revised version of our conference paper \cite{enqvist2022optimizing}. This work significantly expands on the results by including an analysis of the effects of varying pilot power, imperfect CSI, and phase-shift quantization, as well as a detailed analysis of payload length impact.

\section{System model}

We consider a system where a single-antenna UE transmits to a single-antenna BS. There is an RIS in the vicinity of the BS and UE  that can be utilized to improve the channel conditions. We assume that the BS configures and interacts with the RIS.
The question that arises is: when is it worth spending extra radio resources configuring the RIS, if the goal is to minimize the total energy consumption of the UE for UL transmission? 

\begin{figure}[t!]
	\centering 
	\begin{overpic}[width=.99\columnwidth,tics=10]{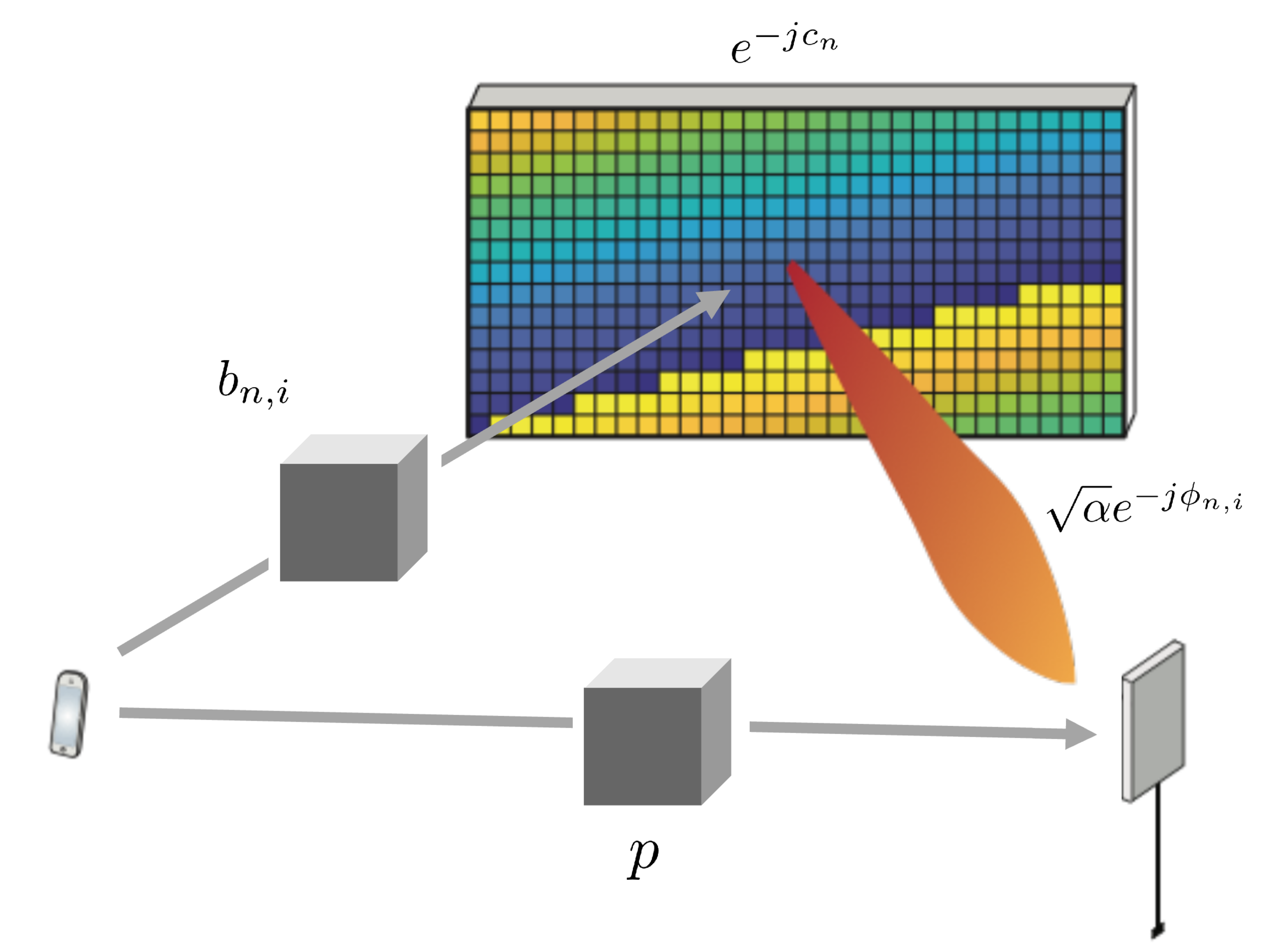}
		\put(2,10.5){\footnotesize UE}
		\put(93,10.5){\footnotesize  BS}
\end{overpic}  \vspace{-2mm}
	\caption{We consider a setup where a UE transmits to a BS through a direct path $p$ and a controllable RIS with $M$ elements divided into subarrays $n=1, \ldots, N$ in each of which the elements range from $i=1, \ldots, M/N$. The channel from the BS to the RIS is LoS whereas the other channels are NLoS, and modeled using Rayleigh fading. The colors represent the phase-shifts of the RIS subarrays.}  \vspace{-3mm}
\label{fig:system-model}
\end{figure}

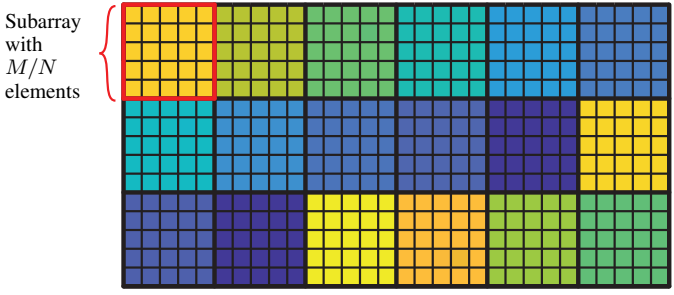
\begin{figure}[t!]
	\centering 
	\begin{overpic}[width=.99\columnwidth,tics=10]{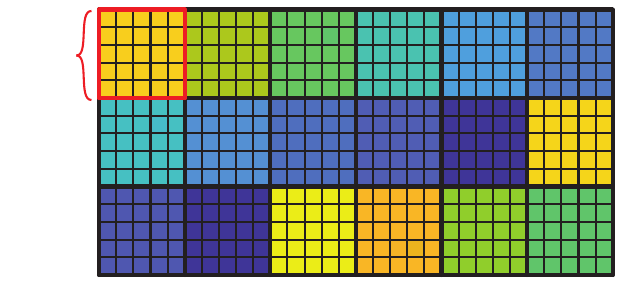}
		\put(-2,40.5){\footnotesize Subarray}
		\put(-2,37){\footnotesize  with}
		\put(-2,33.5){\footnotesize $M/N$}
		\put(-2,30){\footnotesize  elements}
\end{overpic}  \vspace{-2mm}
	\caption{An example setup with an RIS consisting of $M=450$ elements that is divided into $N=18$ subarrays, each with $M/N=25$ elements.}  \vspace{-3mm}
\label{fig:system-model2b}
\end{figure}

We model the communication system as follows. The UE transmits the data signal $x$ with power $P_\mathrm{data}$, i.e., $\mathbb{E}\{\vert x \vert^2\}=P_\mathrm{data}$. The received signal at the BS is
\begin{equation} \label{eq:received-signal}
    y=gx+w,
\end{equation}
where $w\sim \CN (0,\sigma^2)$ is the complex Gaussian thermal receiver noise  and $g$ is the end-to-end channel response, which depends on the RIS configuration and is modeled as follows.

We consider the environment in Fig.~\ref{fig:system-model}, where the BS and RIS are deployed to have a line-of-sight (LoS) channel between them, while all other channels are non-line-of-sight (NLoS) featuring Rayleigh fading. 
The direct path between the BS and UE is modeled as $p\sim \CN(0,\rho)$, where $\rho$ is the average channel gain.

There are also $M$ reflected paths via the $M$ RIS elements. We assume that the RIS is located in the far-field region of the BS and the UE, so that the average channel gain from each RIS element to the BS or the UE is the same.
We assume the RIS elements are divided into $N$ subarrays, each consisting of $M/N$ elements, as exemplified in Fig.~\ref{fig:system-model2b}. 
The RIS assigns the same phase-shift $c_n$ to all the elements in subarray $n$, to limit the channel estimation overhead and reduce the computational complexity.
We will treat $N$ as an optimization variable in this paper and stress that $N=M$ results in a conventional RIS with a different phase-shift at each RIS element. Determining the optimal number of elements, $M$, is important from a design perspective, however, once the RIS has been manufactured, it will not change. On a similar note, once the RIS and the BS have been deployed, the channel gain $\alpha$ between them will be a constant factor.
Since we have LoS paths from the RIS to the BS, all elements have a common gain of $\alpha$ but different phase-shifts. For the RIS element $i\in \{1,\ldots,M/N\}$ in subarray $n$, we let $\phi_{n,i}$ denote the phase-shift of the corresponding path, thus the channel coefficient is $\sqrt{\alpha} e^{-j\phi_{n,i}}$ as shown in Fig.~\ref{fig:system-model}.
The channel from the UE to RIS element $i$ in subarray $n$ is denoted and modeled as $b_{n,i} \sim \CN(0,\beta)$, where $\beta$ is the average channel gain. The channels $b_{n,i}$ and $b_{n',i'}$ are independent for $n\neq n'$ or $i\neq i'$.

In summary, the end-to-end channel response in \eqref{eq:received-signal} is
\begin{equation} \label{eq:end-to-end-channel}
g=
        p+
    \sum_{n=1}^{N} e^{-j c_n}
    \underbrace{\sum_{i=1}^{M/N} \sqrt{\alpha}e^{-j \phi_{n,i}} 
    b_{n,i}}_{=Z_n},
\end{equation}
where we define $Z_n=\sum_{i=1}^{M/N} \sqrt{\alpha}e^{-j \phi_{n,i}}b_{n,i}$ as the concatenated channel of subarray $n$.
We note that $Z_n \sim \CN\left(0,\frac{\alpha \beta M}{N}\right)$ and thus $|Z_n|\sim \mathrm{Rayleigh}\left(\sqrt{
\frac{\alpha \beta M}{2N}}\right)$ is Rayleigh distributed with the scale parameter $\sqrt{
\frac{\alpha \beta M}{2N}}$.
The term $e^{-j c_n}$ represents the common phase-shift $c_n$ introduced by the elements in subarray $n$.\footnote{We have not explicitly included the absorption losses in the RIS elements since these can be absorbed into the channel gains.}
By using \eqref{eq:end-to-end-channel}, we can rewrite \eqref{eq:received-signal} as
\begin{equation} \label{eq:received-signal2}
y=\left(p+\sum_{n=1}^{N} Z_n e^{-j c_n} \right)x+w.
\end{equation}
We will characterize the SNR that can be achieved in this system, depending on the available CSI and phase-shift resolution. We begin with the ideal case.


\subsection{Special case: Perfect CSI and perfect phase-shift resolution}
\label{subsec:perfect}

When the end-to-end channel in \eqref{eq:received-signal2} is used for data transmission, the instantaneous SNR achieved with an arbitrary RIS configuration $c_1,\ldots,c_N$ over the $N$ subarrays is 
\newline
\begin{equation} \label{eq:SNR}
\mathrm{SNR}(c_1,\ldots,c_N) = \ \frac{P_\mathrm{data}}{\sigma^2}\left\vert
p+ \sum_{n=1}^{N} Z_n  e^{-j c_n} \right\vert^2
\end{equation}
\newline
and is maximized by  $c_n=\arg(Z_n)-\arg(p)$  \cite{dinkelbach1967nonlinear}, which gives
\begin{equation} \label{eq:max-SNR}
    \overline{\mathrm{SNR}} = \max_{c_1,\ldots,c_N} \mathrm{SNR}(c_1,\ldots,c_N) = \!
\frac{P_\mathrm{data}}{\sigma^2} \! \left(|p|+ \sum_{n=1}^{N} |Z_n|\right)^2.
\end{equation}

When transmitting a finite-sized data payload, the average SNR determines the likelihood of successful reception (along with the modulation and coding scheme).

Before we derive the exact expression for the average SNR, we will consider a tight lower bound on it, which will be useful when optimizing the SNR and extracting analytical insights. 
A lower bound $\overline{\mathrm{SNR}}_{\rm av,low}$ on the average SNR (with the optimal RIS configuration) can be obtained using Jensen's inequality\footnote{According to Jensen's inequality, for a given random variable $X$ and a convex function of it, $f(X)$, it holds that $\mathbb{E}\left\{ f(X)\right\}\geq f\left(\mathbb{E}\left\{ X\right\}\right)$. We apply Jensen's inequality for the random variable $X=|p|+ \sum_{n=1}^{N} |Z_n|$ and the convex function $f(X)=X^2$.} as
\begin{align}
\label{eq:Jensen}
    \mathbb{E} \left\{\overline{\mathrm{SNR}}\right\} & \geq \overline{\mathrm{SNR}}_{\rm av,low} = \frac{P_\mathrm{data}}{\sigma^2} \left(\mathbb{E}\left\{|p|+ \sum_{n=1}^{N} |Z_n|\right\}\right)^2 \nonumber \\
    &\stackrel{(a)}{=} \frac{P_\mathrm{data}}{\sigma^2} \left(\frac{\sqrt{\pi}}{2}\sqrt{\rho}+  N  \frac{\sqrt{\pi}}{2}\sqrt{
\frac{\alpha \beta M}{N}}  \right)^2 \nonumber \\
&=\frac{\pi P_\mathrm{data}}{4\sigma^2} 
\left(\sqrt{\rho}+\sqrt{\alpha \beta M N}  \right)^2,
\end{align}
where the equality in $(a)$ follows from the fact that $|p|\sim \mathrm{Rayleigh}\left(\sqrt{
\frac{\rho}{2}}\right)$ and $E\{|p|\}=\frac{\sqrt{\pi}}{2}\sqrt{\rho}$. Similarly, $|Z_n|\sim \mathrm{Rayleigh}\left(\sqrt{
\frac{\alpha \beta M}{2N}}\right)$ and $E\left\{|Z_n|\right\} = \frac{\sqrt{\pi}}{2}\sqrt{
\frac{\alpha \beta M}{N}}$.

The following lemma gives the exact average SNR.

\begin{lemma}
With an optimized RIS configuration, the average of the SNR in \eqref{eq:max-SNR} is
\begin{equation} \label{eq:average-SNR}
\mathbb{E}\left\{\overline{\mathrm{SNR}}\right\} = \overline{\mathrm{SNR}}_{\rm av,low} + \frac{ P_\mathrm{data}}{ \sigma^2}\left(1-\frac{\pi}{4}\right)\left(\rho+\alpha\beta M\right) .
\end{equation}
\end{lemma}

\begin{IEEEproof}
By expanding the square in \eqref{eq:max-SNR}, we obtain 
\begin{align} \begin{split}
\mathbb{E} \left\{\overline{\mathrm{SNR}}\right\}=&
\frac{P_\mathrm{data}}{\sigma^2}\Bigg(\mathbb{E}\left\{|p|^2\right\}+2\mathbb{E}\left\{|p|\right\}\sum_{n=1}^{N} \mathbb{E}\left\{|Z_n|\right\} \\
&+\sum_{n=1}^{N} \sum_{m=1}^{N}\mathbb{E}\left\{|Z_n||Z_m|\right\}\Bigg), \end{split}
\end{align}
where we have used the independence of $Z_n$ from $p$. Moreover, noting that $\mathbb{E}\left\{|Z_n||Z_m|\right\}=\mathbb{E}\left\{|Z_n|\right\}\mathbb{E}\left\{|Z_m|\right\}$ for $n\neq m$, due to the independence of $Z_n$ and $Z_m$, we can simplify $\mathbb{E} \left\{\overline{\mathrm{SNR}}\right\}$ as
\begin{align}
&\frac{P_\mathrm{data}}{\sigma^2} \Bigg(\rho+\frac{\pi}{2}\sqrt{\rho\alpha\beta M N} + N \frac{\alpha \beta M}{N} \nonumber\\
&\hspace{12mm}+N(N-1)\left( \frac{\sqrt{\pi}}{2} \sqrt{\frac{\alpha\beta M}{N}}\right)^2 \Bigg) \nonumber \\
&= \frac{ P_\mathrm{data}}{ \sigma^2} \Bigg( \frac{\pi}{4}
\left(\sqrt{\rho}+\sqrt{\alpha \beta M N}\right)^2
+\left(1-\frac{\pi}{4}\right) \!(\rho+\alpha \beta M) \!
\Bigg),   
\label{eq:exactSNR}
\end{align}
where we have used that $E\{|p|\}=\frac{\sqrt{\pi}}{2}\sqrt{\rho}$ and $E\left\{|Z_n|\right\} = \frac{\sqrt{\pi}}{2}\sqrt{
\frac{\alpha \beta M}{N}}$. We finally obtain \eqref{eq:average-SNR} by recalling the expression of the lower bound $\overline{\mathrm{SNR}}_{\rm av,low}$ in \eqref{eq:Jensen}.

\end{IEEEproof}

The exact average SNR in \eqref{eq:average-SNR} is an increasing function of the channel gains $\rho$, $\alpha$, $\beta$, the number of RIS elements $M$, and the number of subarrays $N$.
We further note that \eqref{eq:average-SNR} is a summation of its lower bound in \eqref{eq:Jensen} and a term that is independent of $N$. Hence, for a given $M$, the number of subarrays $N$ maximizing the lower bound of the average SNR also maximizes the exact average SNR. 
$\overline{\mathrm{SNR}}_{\rm av,low}$ is always greater than the additional term in \eqref{eq:average-SNR}, which is often negligible. To demonstrate this, Fig.~\ref{fig:ESNR} compares the exact SNR expression in \eqref{eq:average-SNR} and its lower bound in \eqref{eq:Jensen}. We consider a setup where the channel gains are
$\alpha=-80$\,dB, $\beta=-60$\,dB, and $\rho=-95$\,dB. The RIS consists of $M=1024$ elements, which can be divided into $N=2^R$ number of subarrays for $R=0,1, \dotsc ,10$. The ``transmit SNR'' is $P_\mathrm{data}/\sigma^2=104$\,dB. Also shown is a comparison to not using subarrays but instead $N$ individually configured elements. In this baseline case, where $M-N$ elements are turned off as in prior work \cite{9244106,9333612}, we obtain the average maximum SNR by inserting $N$ in place of $M$ in \eqref{eq:exactSNR} as
\begin{align}
\begin{split}    
\mathbb{E}\left\{\overline{\mathrm{SNR}}_b\right\}=&\hspace{2mm}\frac{ P_\mathrm{data}}{ \sigma^2} \Bigg(
\frac{\pi}{4}\left(\sqrt{\rho}+N\sqrt{\alpha \beta}\right)^2 \\
&+\left(1-\frac{\pi}{4}\right) \!(\rho+\alpha \beta N) \!
\Bigg).\end{split}
\end{align}
It is clearly seen that
$\mathbb{E}\left\{\overline{\mathrm{SNR}}_b\right\} \leq \mathbb{E}\left\{\overline{\mathrm{SNR}}\right\}$
with equality only if $N=M$. This is also demonstrated in Fig.~\ref{fig:ESNR}. Furthermore, as $N$ increases, the gap between the exact SNR expression in \eqref{eq:average-SNR} and its lower bound in \eqref{eq:Jensen} decreases as also verified in the figure. The black stars are the average from 10\,000 Monte Carlo realizations, confirming the validity of the formulas.

\begin{figure}[t!] \centering
\centerline{\includegraphics[trim=8 2 25 15,clip,width=0.99\columnwidth]{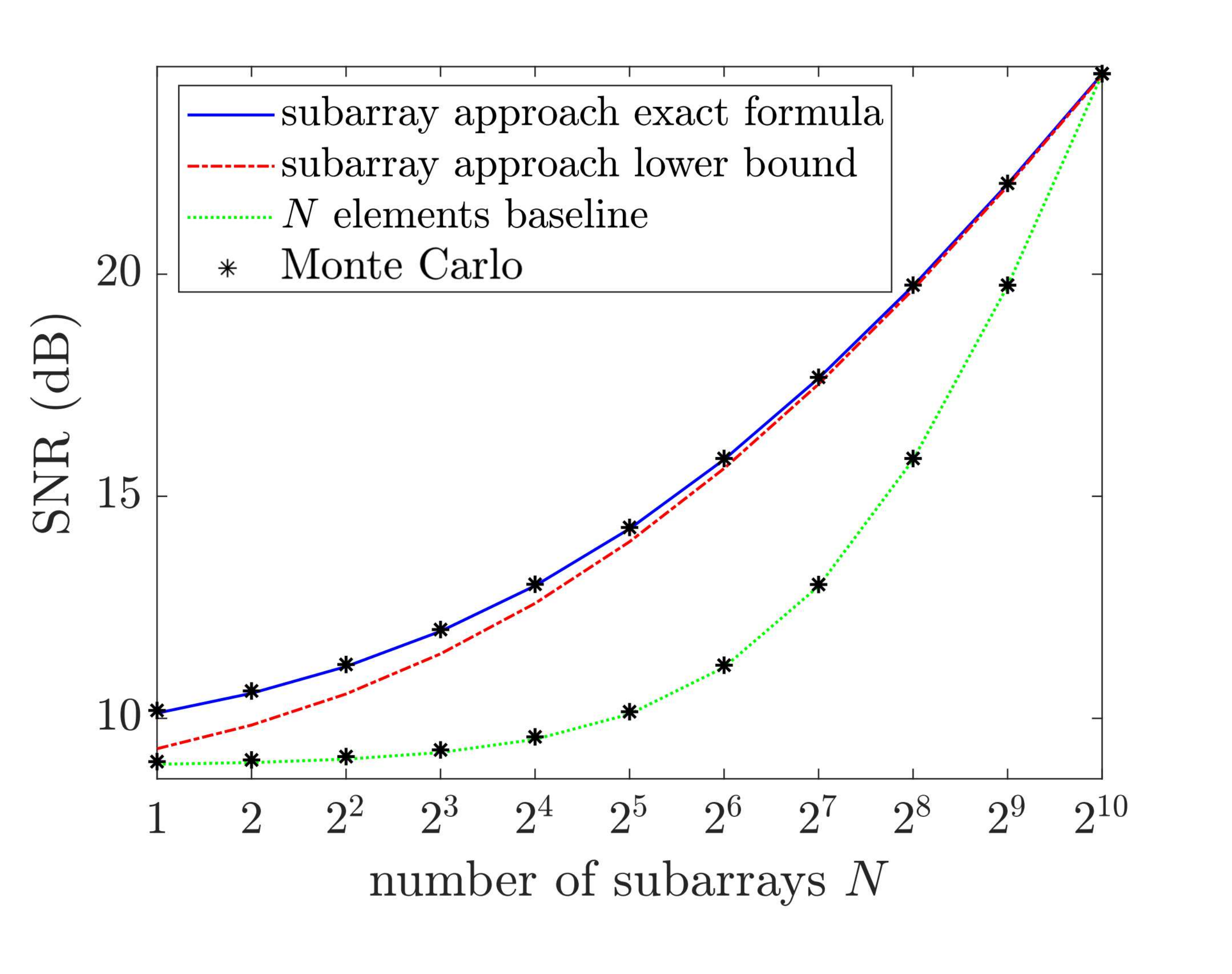}}  
\caption{The average SNR versus $N$ with an optimized RIS configuration for $M=1024$ RIS elements being divided into $N=2^R$ subarrays for $R=0,1, \dotsc ,10$. The baseline with only $N$ individually controlled elements is also shown. {\color{black}The black stars are computed as averages from 10\,000 Monte Carlo simulations.}
} \label{fig:ESNR}
\end{figure}

\subsection{Characterizing the SNR during pilot transmission}

In the remainder of this section, we will extend the SNR analysis to cases with imperfect CSI. It is then important to distinguish between the SNR during pilot transmission (before the RIS has been configured) and the SNR during data transmission (after the RIS has been configured). We begin by characterizing the SNR during pilot transmission.

To determine the RIS configuration that maximizes the data-transmission SNR in \eqref{eq:max-SNR}, we need to learn the current channel realization. 
When using $N$ subarrays, there are effectively $N+1$ channel coefficients in  \eqref{eq:received-signal2} to estimate; one for each subarray with RIS elements plus the direct path. Let 
\begin{align} \label{eq:channel-h}
\textbf{h} =    \begin{bmatrix} p & Z_1 & \hdots & Z_{N} \end{bmatrix}^{\Ttran} \in \mathbb{C}^{N+1}
\end{align}
denote the vector containing these channel coefficients from \eqref{eq:received-signal2}.
The benefit of using subarrays is apparent in that the number of coefficients to be estimated is lower than the total number of RIS elements.
Since there are $N+1$ unknown coefficients, we assume the UE sends a predefined pilot sequence of length $N+1$  to obtain sufficient degrees-of-freedom to estimate the individual coefficients in $\textbf{h}$. 
Without loss of generality, the pilot sequence is constant $\sqrt{P_\mathrm{pilot}}[1 \ \ldots \ 1 ]^{\Ttran}\in \mathbb{C}^{N+1}$, where $P_\mathrm{pilot}$ is the transmit power.
During the pilot transmission, the RIS changes its configuration according to a structured deterministic phase-shift pattern. Let $\bm{\psi}_t\in \mathbb{C}^{N+1}$ be the vector with $1$ as its first entry (representing the direct path) and the configured phase-shifts $e^{-j\psi_{t,n}}$ for subarrays $n=1,\ldots,N$ during the $t$th pilot  transmission as the other entries. This vector is given as
\begin{align}
   \bm{\psi}_t = \begin{bmatrix} 1 & e^{-j\psi_{t,1}} & \hdots & e^{-j\psi_{t,N}} \end{bmatrix}^{\Ttran}, \quad t=1,\ldots,N+1.
\end{align}
 Then, the received signal at the BS during the pilot transmission is given in vector form as
\begin{align}
    \textbf{y}_\mathrm{pilot} = \sqrt{P_\mathrm{pilot}}\underbrace{\begin{bmatrix} \bm{\psi}_1^{\Ttran} \\ \vdots \\ \bm{\psi}_{N+1}^{\Ttran}  \end{bmatrix}}_{\triangleq \boldsymbol{\Psi}}\textbf{h}+ \textbf{w}_\mathrm{pilot},
\end{align}
where $\textbf{w}_\mathrm{pilot}\sim \CN(\textbf{0},\sigma^2\textbf{I}_{N+1})$ is the thermal receiver noise. If we select the matrix $\boldsymbol{\Psi}$ as any scaled unitary matrix (e.g., a discrete Fourier transform matrix), then $\boldsymbol{\Psi}^{\Htran}\boldsymbol{\Psi}=(N+1)\textbf{I}_{N+1}$ and the BS can estimate the individual paths as
\begin{align} 
\label{eq:unitarypilot}
 \underbrace{\frac{\boldsymbol{\Psi}^{\Htran}\textbf{y}_\mathrm{pilot}}{\sqrt{N+1}}}_{\triangleq \overline{\textbf{y}}_{\mathrm{pilot}}}  = \sqrt{(N+1)P_\mathrm{pilot}}\textbf{h} +  \underbrace{\frac{\boldsymbol{\Psi}^{\Htran}\textbf{w}_\mathrm{pilot}}{\sqrt{N+1}}}_{\triangleq \overline{\textbf{w}}_{\mathrm{pilot}}} ,
\end{align}
where the effective noise $\overline{\textbf{w}}_{\mathrm{pilot}}$ has the same distribution as $\textbf{w}_\mathrm{pilot}$ since $\boldsymbol{\Psi}^{\Htran}/\sqrt{N+1}$ is a unitary matrix \cite{massivemimobook}. As the entries of $\overline{\textbf{y}}_{\mathrm{pilot}}$ and $\overline{\textbf{w}}_{\mathrm{pilot}}$ are independent, the minimum mean square error (MMSE) estimates of the entries of $\textbf{h}$ from \eqref{eq:channel-h} can be obtained as follows:
\begin{align}
    \hat{p} =& \frac{\mathbb{E}\left\{ p \overline{y}_0^*\right\}}{\mathbb{E}\left\{|\overline{y}_0|^2\right\}} \overline{y}_0 = \frac{\sqrt{(N+1)P_{\mathrm{pilot}}}\mathbb{E}\{|p|^2\}}{(N+1)P_{\mathrm{pilot}}\mathbb{E}\{|p|^2\}+\sigma^2}\overline{y}_0\nonumber\\
=&\frac{\sqrt{(N+1)P_{\mathrm{pilot}}}\rho}{(N+1)P_{\mathrm{pilot}}\rho+\sigma^2}\overline{y}_0, \label{eq:hatp}\\
  \hat{Z}_n =& \frac{\mathbb{E}\left\{ Z_n \overline{y}_n^*\right\}}{\mathbb{E}\left\{|\overline{y}_n|^2\right\}} \overline{y}_n = \frac{\sqrt{(N+1)P_{\mathrm{pilot}}}\mathbb{E}\{|Z_n|^2\}}{(N+1)P_{\mathrm{pilot}}\mathbb{E}\{|Z_n|^2\}+\sigma^2}\overline{y}_n\nonumber\\
=&\frac{\sqrt{(N+1)P_{\mathrm{pilot}}}\frac{\alpha\beta M}{N}}{(N+1)P_{\mathrm{pilot}}\frac{\alpha\beta M}{N}+\sigma^2}\overline{y}_n, \label{eq:hatZn}
\end{align}
for $n=1,\ldots,N$, where
\begin{align}
\overline{\textbf{y}}_{\mathrm{pilot}}=\begin{bmatrix} \overline{y}_0 & \overline{y}_1 & \hdots & \overline{y}_N \end{bmatrix}^{\Ttran} \in \mathbb{C}^{N+1}.
\end{align}

The channel estimates have zero mean and their variances can be obtained as
\begin{align}
\label{eq:p_hat}
&\mathbb{E}\left\{\left|\hat{p}\right|^2\right\} = \frac{(N+1)P_{\mathrm{pilot}}\rho^2}{(N+1)P_{\mathrm{pilot}}\rho+\sigma^2},\\
& \label{eq:z_hat}\mathbb{E}\left\{\left|\hat{Z}_n\right|^2\right\} = \frac{(N+1)P_{\mathrm{pilot}}\left(\frac{\alpha\beta M}{N}\right)^2}{(N+1)P_{\mathrm{pilot}}\frac{\alpha\beta M}{N}+\sigma^2}, \quad n=1,\ldots,N.
\end{align}

Recalling the entries of the channel vector $\textbf{h}$ from \eqref{eq:channel-h}, we have different pilot SNR for the direct BS-UE path and the paths through the $N$ subarrays of RIS. Using $\mathbb{E}\left\{|p|^2\right\}=\rho$ and $\mathbb{E}\left\{|Z_n|^2\right\}=\frac{\alpha\beta M}{N}$, the pilot SNR for the direct channel and the cascaded RIS channels are given as 
\begin{align} \label{eq:SNR_pilot}
\mathrm{SNR}_\mathrm{pilot} = \begin{cases} \frac{(N+1)P_\mathrm{pilot}\rho}{\sigma^2} & \text{for the direct channel}, \\
\frac{(N+1)P_\mathrm{pilot}\alpha\beta M}{N\sigma^2} & \text{for the RIS channels}.\end{cases}
\end{align}


\subsection{General case: Phase-shift quantization and imperfect CSI}

When operating the RIS in practical conditions, the CSI will not only be imperfect but also the phase-shift resolution will be limited. There are two main reasons for this. Firstly, the control information between the RIS and the BS must have a finite resolution. Secondly, the number of distinct impedance values that the RIS elements support (e.g., by selecting different bias voltages) is often limited in practical implementations \cite{sayanskiy20222d}. In this section, we analyze how the SNR is affected when the RIS subarrays (or individual elements) are only allowed to have phase-shifts selected from a discrete set of values.

We consider that the phase-shifts at each element of the RIS can only take a finite number
of discrete values, which are equally spaced in [0, $2\pi$). Let $b$ denote the number of bits used to represent each of the
phase-shift values. The quantized phase-shift configurations with
either $b=1$, $b=2$, or $b=3$ bits of resolution corresponding to either $2^1=2$, $2^2=4$, or $2^3=8$ different possible reflection coefficients are illustrated in Fig.~\ref{fig:quantSNRR}, where $\hat{c}_n$ denotes the quantized phase-shift introduced by subarray $n$. The possible quantized phase-shifts for a 1-bit RIS are $+\frac{\pi}{4}$ and $-\frac{3\pi}{4}$. For a 2-bit RIS, the four possible quantized phase-shift values are $\pm \frac{\pi}{4}$ and $\pm \frac{3\pi}{4}$, and for a 3-bit RIS, the quantized phase-shifts are $0$, $\pi$,  $\pm \frac{\pi}{2}$, $\pm \frac{\pi}{4}$, and $\pm \frac{3\pi}{4}$. The results in this paper are not limited to these specific phase-shift values but only utilize the fact that the phases are equally spaced on the unit circle.

\begin{figure}[t!]
\centerline{\includegraphics[width=0.88\columnwidth]{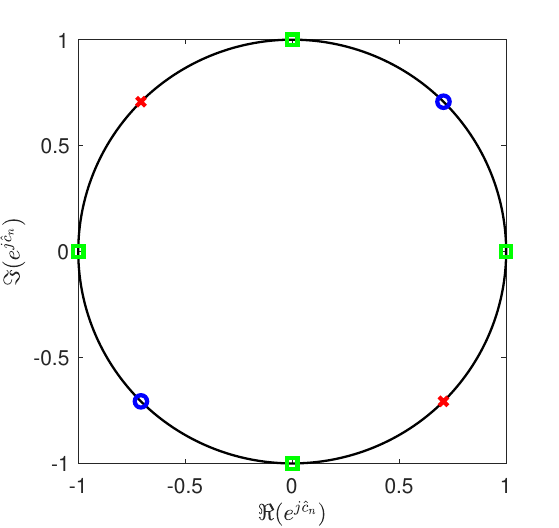}} 
\caption{The quantized phase-shift configurations. The two blue circular points are available for a 1-bit RIS. The blue points and the points marked by red crosses are available when using a 2-bit RIS. All the eight points are available for a 3-bit RIS. For an RIS with ``perfect phase-shift resolution'' as in Section~\ref{subsec:perfect}, any point on the unit circle may be utilized.}  \label{fig:quantSNRR}
\end{figure}


The adverse effect of the finite quantization resolution on the SNR as a function of the number of subarrays can be seen in Fig.~\ref{fig:quantSNR}. In this setup, $P_\mathrm{data}/\sigma^2=104$\,dB and we assume perfect CSI and focus on a setup without a direct path, i.e., $p=0$. The other channel gains are $\alpha=-80$\,dB and $\beta=-60$\,dB.  $M=1024$ RIS elements are divided into subarray sizes of $N=2^R$ for $R=0,1, \dotsc ,10$. The baseline of no quantization, which stands for infinite phase-shift resolution, is also shown. Later, we will quantify the SNR gap between the quantized phase-shift case and the no quantization in Lemma~\ref{lem:upperbound-SNR-loss}. For now, it is observed from the figure that using three bits for phase-shifts almost achieves the same SNR as the perfect resolution does. The SNR loss due to using two bits is also small.

When we only have access to imperfect CSI, the RIS phase-shifts for each subarray $n$ would ideally (when there is no quantization) be obtained as $c_n=\arg(\hat{Z}_n)-\arg(\hat{p})$ using the channel estimates from \eqref{eq:hatp} and \eqref{eq:hatZn}. This is in line with maximizing the signal strength of the known part of the end-to-end channel in \eqref{eq:max-SNR}. Accounting for the combined effect of the imperfect CSI and RIS phase-shift quantization, we can write: $\hat{c}_n= c_n+\tilde{c}_n=\arg(\hat{Z}_n)-\arg(\hat{p})+\tilde{c}_n$.  We have $\hat{c}_n = Q_b \left(e^{j\arg(\hat{Z}_n)}e^{-j\arg\left(\hat{p}\right)}\right) =\arg(\hat{Z}_n)-\arg(\hat{p})  + \tilde{c}_n$, where $Q_b(\cdot)$ is the $b$-bit quantization operator that maps the complex input to the nearest quantized phase-shift value on the unit circle as shown in Fig.~\ref{fig:quantSNRR}. When the quantized values of the phase-shifts maximizing the signal strength of the known part of the end-to-end channel are introduced by the RIS, the SNR becomes
\begin{equation} \label{eq:SNR-imperfectCSI}
\mathrm{SNR}(c_1,\ldots,c_N) = \ \frac{P_\mathrm{data}}{\sigma^2}\left\vert
p+ \sum_{n=1}^{N} Z_n  e^{-j \hat{c}_n} \right\vert^2.
\end{equation}

 The probability distribution of the quantization error $\tilde{c}_n$ can be obtained by noting the circular symmetry of the quantization operator. It is a fact that the quantization error for any phase-shift value is uniformly distributed on $\left[-\pi/K,\pi/K\right)$ for $K=2^b$ equally-spaced quantization points on the unit circle. We denote this by $\tilde{c}_n \sim \mathcal{U}\left[-\pi/K,\pi/K\right)$. By utilizing the moment generating function of $\tilde{c}_n$, we obtain a result that will be needed in the following derivation. The moment generating function of $\tilde{c}_n$ is given as
\begin{align}
    \mathbb{E}\left\{ e^{j \tilde{c}_n}\right \}&=\int_{-\pi/K}^{\pi/K} \frac{1}{2\pi/K}e^{j u}du= \frac{e^{j\pi/K}-e^{-j \pi/K}}{j(2\pi/K)}\nonumber\\
    &=\sinc (1/K) \label{eq:MGF}
\end{align}
where $\sinc(x)=\sin(\pi x)/(\pi x)$ is the sinc-function. 
That is, the implemented reflection coefficient is equal to the desired phase-shift obtained by the channel estimates, i.e., $c_n=\arg(\hat{Z}_n)-\arg\left(\hat{p}\right)$, plus the uniformly distributed quantization error, i.e., $\tilde{c}_n=Q_b(e^{jc_n})-c_n$. This leads us to our first main result that takes into account the combined effect of quantization and imperfect CSI.

\begin{theorem} \label{th:general}
When the phase-shifts of the RIS elements in subarray $n$ are selected as $\hat{c}_n = Q_b \left(e^{j\arg(\hat{Z}_n)}e^{-j\arg\left(\hat{p}\right)}\right)$ with $2^b=K$ equally-spaced quantization points on the unit circle, the average value of the SNR in \eqref{eq:SNR-imperfectCSI} becomes
\begin{align} 
\begin{split}
&\mathbb{E}\left\{  \overline{\mathrm{SNR}}\right\}=\frac{P_\mathrm{data}}{\sigma^2}\Bigg(\rho+\alpha \beta M \\
&+ \frac{\pi}{2}\frac{\rho \alpha \beta M P_{\mathrm{pilot}}/\sigma^2 \sinc(1/K)}{\sqrt{\left(\rho P_{\mathrm{pilot}}/\sigma^2+\frac{1}{N+1}\right)\left(\frac{\alpha \beta M}{N} P_{\mathrm{pilot}}/\sigma^2+ \frac{1}{N+1}\right)}}\\
    &+\frac{\pi}{4}\left(1-\frac{1}{N}\right)\frac{(\alpha \beta M)^2P_{\mathrm{pilot}}/\sigma^2 \sinc^2(1/K)}{\frac{\alpha \beta M}{N} P_{\mathrm{pilot}}/\sigma^2+\frac{1}{N+1}}\Bigg). \label{eq:lemma2}
    \end{split}
    \end{align}
\end{theorem}

\begin{IEEEproof}
By expanding the square in \eqref{eq:SNR-imperfectCSI} we obtain
\begin{align}
    &\mathbb{E}\left\{\overline{\mathrm{SNR}} \right\}=\frac{P_\mathrm{data}}{\sigma^2}\mathbb{E}\left\{\left\vert p+
\sum_{n=1}^{N}  Z_n  e^{-j \hat{c}_n} \right\vert^2 \right\} \nonumber \\
&= \frac{P_\mathrm{data}}{\sigma^2}\mathbb{E}\left\{\left( p +\sum_{n=1}^{N} Z_n  e^{-j \hat{c}_n}   \right)\left( p^* +\sum_{m=1}^{N} Z_m^*  e^{j \hat{c}_m}   \right)\right\} \nonumber \\
&=  \frac{P_\mathrm{data}}{\sigma^2}\mathbb{E}\Bigg\{|p|^2+p\sum_{m=1}^{N} Z_m^*  e^{j \hat{c}_m} +p^*\sum_{n=1}^{N} Z_n  e^{-j \hat{c}_n}  \nonumber \\
&\hspace{17mm}+\sum_{n=1}^{N} \sum_{m=1}^{N}Z_n  e^{-j \hat{c}_n}Z_m^*  e^{j \hat{c}_m}\Bigg\} \nonumber \\
&\stackrel{(a)}{=}\frac{P_\mathrm{data}}{\sigma^2}\Bigg(\underbrace{\mathbb{E}\left\{|p|^2\right\}}_{=\rho} 
+N\underbrace{\mathbb{E}\left\{pZ_n^*e^{j \hat{c}_n}\right\}}_{=\mathcal{A}} \nonumber \\
&\hspace{17mm}+N\underbrace{\mathbb{E}\left\{p^*Z_n e^{-j \hat{c}_n}\right\}}_{=\mathcal{A}^*}+N\underbrace{\mathbb{E}\left\{ |Z_n|^2\right\}}_{=\frac{\alpha\beta M}{N}}\nonumber\\
&\hspace{17mm}+\sum_{n=1}^N\condSum{m=1}{m\neq n}{N}\underbrace{\mathbb{E}\left\{Z_n Z_m^* e^{-j (\hat{c}_n - \hat{c}_m)}\right\}}_{=\mathcal{B}}\Bigg) \nonumber \\
&= \frac{P_\mathrm{data}}{\sigma^2}\big( \rho + 2N\Re\{\mathcal{A}\} +\alpha \beta M+N(N-1)\mathcal{B}\big) \label{eq:result-lemma2},
\end{align}
where we utilized $\mathbb{E}\left\{|p|^2\right\}=\rho$ and $\mathbb{E}\left\{|Z_n|^2\right\}=\frac{\alpha\beta M}{N}$ in $(a)$. Moreover, we have defined $\mathcal{A}=\mathbb{E}\left\{pZ_n^*e^{j \hat{c}_n}\right\}$, $\forall n$ and $\mathcal{B}=\mathbb{E}\left\{Z_n Z_m^* e^{-j (\hat{c}_n - \hat{c}_m)}\right\}$, $\forall n\neq m$. We compute $\mathcal{A}$ and $\mathcal{B}$ as follows.

\emph{1) Computation of $\mathcal{A}$.}

Noting that $Z_n=\hat{Z}_n+\tilde{Z}_n$, $p=\hat{p}+\tilde{p}$, where $\tilde{Z}_n$ and $\tilde{p}$   are the channel estimation errors, and $\hat{c}_n= \arg(\hat{Z}_n)-\arg\left(\hat{p}\right)+\tilde{c}_n$, we express $\mathcal{A}$ as
\begin{align}
    \mathcal{A}&=\mathbb{E}\left\{pZ_n^*e^{j \hat{c}_n}\right\} \nonumber \\
    &=\mathbb{E}\left\{ \left(\hat{p}+\tilde{p}\right)\left(\hat{Z}_n^*+\tilde{Z}_n^*\right)e^{j\arg\left(\hat{Z}_n\right)}e^{-j\arg\left(\hat{p}\right)}e^{j\tilde{c}_n}\right\} \nonumber \\
    &=\mathbb{E}\Big\{ \left(\hat{p}e^{-j\arg\left(\hat{p}\right)}+\tilde{p}e^{-j\arg\left(\hat{p}\right)}\right) \nonumber \\ &\hspace{9mm}\times\left(\hat{Z}_n^*e^{j\arg\left(\hat{Z}_n\right)}+\tilde{Z}_n^*e^{j\arg\left(\hat{Z}_n\right)}\right)e^{j\tilde{c}_n}\Big\} \nonumber \\
    &=\mathbb{E}\left\{ \left(\left|\hat{p}\right|+\tilde{p}e^{-j\arg\left(\hat{p}\right)}\right)\left(\left|\hat{Z}_n\right|+\tilde{Z}_n^*e^{j\arg\left(\hat{Z}_n\right)}\right)e^{j\tilde{c}_n}\right\} \nonumber\\
    &= \mathbb{E}\left\{\left|\hat{p}\right|\left|\hat{Z}_n\right|e^{j\tilde{c}_n}\right\} + \mathbb{E}\left\{\left|\hat{p}\right|\tilde{Z}_n^*e^{j\arg\left(\hat{Z}_n\right)}e^{j\tilde{c}_n}\right\}  \nonumber\\
    &\hspace{4mm}+ \mathbb{E}\left\{\tilde{p}e^{-j\arg\left(\hat{p}\right)}\left|\hat{Z}_n\right|e^{j\tilde{c}_n}\right\} \nonumber \\
    &\hspace{4mm}+ \mathbb{E}\left\{\tilde{p}e^{-j\arg\left(\hat{p}\right)}\tilde{Z}_n^*e^{j\arg\left(\hat{Z}_n\right)}e^{j\tilde{c}_n}\right\}  \nonumber \\
    &\stackrel{(a)}{=}\mathbb{E}\left\{\left|\hat{p}\right|\left|\hat{Z}_n\right|e^{j\tilde{c}_n}\right\} + \underbrace{\mathbb{E}\left\{\tilde{Z}_n^*\right\}}_{=0}\mathbb{E}\left\{\left|\hat{p}\right|e^{j\arg\left(\hat{Z}_n\right)}e^{j\tilde{c}_n}\right\} \nonumber \\
    &\hspace{4mm}+ \underbrace{\mathbb{E}\left\{\tilde{p}\right\}}_{=0}\mathbb{E}\left\{e^{-j\arg\left(\hat{p}\right)}\left|\hat{Z}_n\right|e^{j\tilde{c}_n}\right\} \nonumber \\
    &\hspace{4mm}+ \underbrace{\mathbb{E}\left\{\tilde{p}\right\}}_{=0}\underbrace{\mathbb{E}\left\{\tilde{Z}_n^*\right\}}_{=0}\mathbb{E}\left\{e^{-j\arg\left(\hat{p}\right)}
e^{j\arg\left(\hat{Z}_n\right)}e^{j\tilde{c}_n}\right\}\nonumber\\
     &\stackrel{(b)}{=}\mathbb{E}\left\{\left|\hat{p}\right|\right\}\mathbb{E}\left\{\left|\hat{Z}_n\right|\right\}\mathbb{E}\left\{e^{j\tilde{c}_n}\right\} ,
\end{align}
where we used the independence of the channel estimation errors $\tilde{Z}_n$ and $\tilde{p}$ from each other and from $\hat{Z}_n$ and $\hat{p}$ in $(a)$. Since $c_n$, and thus $\tilde{c}_n$  is a function of $\hat{Z}_n$ and $\hat{p}$, the channel estimation errors $\tilde{Z}_n$ and $\tilde{p}$ are also independent of $\tilde{c}_n$. We also utilize that the channel estimation errors have zero mean. In $(b)$, we have used that the random variables $|\hat{p}|$, $|\hat{Z}_n|$, and $\tilde{c}_n$ are mutually independent. This follows from the fact that for the complex Gaussian random variable $\hat{p}$, $|\hat{p}|$ and $\arg\left(\hat{p}\right)$ are independent. Similarly, for the complex Gaussian random variable $\hat{Z}_n$, $|\hat{Z}_n|$ and $\arg(\hat{Z}_n)$ are independent. Furthermore,  $\tilde{c}_n$ is a function of $\arg(\hat{p})$ and $\arg(\hat{Z}_n)$, which leads to the independence of $\tilde{c}_n$ from $|\hat{p}|$ and $|\hat{Z}_n|$. Noting that $|\hat{p}|$ and $|\hat{Z}_n|$ are Rayleigh distributed and using \eqref{eq:p_hat}-\eqref{eq:z_hat}, we have 

\begin{align}
&\mathbb{E}\left\{\left|\hat{p}\right|\right\} = \frac{\sqrt{\pi}}{2}\sqrt{\frac{(N+1)P_{\mathrm{pilot}}\rho^2}{(N+1)P_{\mathrm{pilot}}\rho+\sigma^2}}\\
&\mathbb{E}\left\{\left|\hat{Z}_n\right|\right\} = \frac{\sqrt{\pi}}{2}\sqrt{\frac{(N+1)P_{\mathrm{pilot}}\left(\frac{\alpha\beta M}{N}\right)^2}{(N+1)P_{\mathrm{pilot}}\frac{\alpha\beta M}{N}+\sigma^2}}, \quad n=1,\ldots,N. \label{eq:exp-abs-Zhat}
\end{align}

Finally, recalling \eqref{eq:MGF} and arranging the terms, we obtain the closed-form expression of $\mathcal{A}$ as

\begin{align}
    \mathcal{A}=\frac{\pi}{4\sqrt{N}}\frac{\rho \alpha \beta M P_{\mathrm{pilot}}/\sigma^2 \sinc(1/K)}{\sqrt{\left(\rho P_{\mathrm{pilot}}/\sigma^2+\frac{1}{N+1}\right)\left(\alpha \beta M P_{\mathrm{pilot}}/\sigma^2+ \frac{N}{N+1}\right)}}. \label{eq:mathcalA}
\end{align}

\emph{2) Computation of $\mathcal{B}$}

Noting that $Z_n=\hat{Z}_n+\tilde{Z}_n$ and $\hat{c}_n= \arg(\hat{Z}_n)-\arg\left(\hat{p}\right)+\tilde{c}_n$, we express $\mathcal{B}$ for $n\neq m$ as

\begin{align*}
   \mathcal{B}= &\mathbb{E}\left\{Z_n Z_m^* e^{-j (\hat{c}_n - \hat{c}_m)}\right\} = \mathbb{E}\Big\{\left(\hat{Z}_n+\tilde{Z}_n\right)\left(\hat{Z}_m^*+\tilde{Z}_m^*\right) \nonumber \\
   &\hspace{14mm}\times e^{-j\arg\left(\hat{Z}_n\right)} e^{-j\tilde{c}_n}e^{j\arg\left(\hat{Z}_m\right)}e^{j\tilde{c}_m}\Big\} \nonumber \\
    =& \mathbb{E}\Big\{\left(\left|\hat{Z}_n\right|+\tilde{Z}_ne^{-j\arg\left(\hat{Z}_n\right)}\right) \nonumber \\
&\hspace{14mm}\times\left(\left|\hat{Z}_m\right|+\tilde{Z}_m^*e^{j\arg\left(\hat{Z}_m\right)}\right)e^{-j\tilde{c}_n}e^{j\tilde{c}_m}\Big\} \nonumber \\
    =&\mathbb{E}\left\{ \left|\hat{Z}_n\right|\left|\hat{Z}_m\right|e^{-j\tilde{c}_n}e^{j\tilde{c}_m}\right\} \nonumber \\
    &+\mathbb{E}\left\{\tilde{Z}_ne^{-j\arg\left(\hat{Z}_n\right)}\left|\hat{Z}_m\right|e^{-j\tilde{c}_n}e^{j\tilde{c}_m}\right\} \nonumber \\
    &+ \mathbb{E}\left\{\left|\hat{Z}_n\right|\tilde{Z}_m^*e^{j\arg\left(\hat{Z}_m\right)}e^{-j\tilde{c}_n}e^{j\tilde{c}_m}\right\} \nonumber \\
    &+\mathbb{E}\left\{\tilde{Z}_ne^{-j\arg\left(\hat{Z}_n\right)}\tilde{Z}_m^*e^{j\arg\left(\hat{Z}_m\right)}e^{-j\tilde{c}_n}e^{j\tilde{c}_m}\right\}\nonumber \\
    \end{align*}
    \begin{align}
    \stackrel{(a)}{=}& \mathbb{E}\left\{ \left|\hat{Z}_n\right|\left|\hat{Z}_m\right|e^{-j\tilde{c}_n}e^{j\tilde{c}_m}\right\} \nonumber \\
    &+\underbrace{\mathbb{E}\left\{\tilde{Z}_n\right\}}_{=0}\mathbb{E}\left\{e^{-j\arg\left(\hat{Z}_n\right)}\left|\hat{Z}_m\right|e^{-j\tilde{c}_n}e^{j\tilde{c}_m}\right\} \nonumber \\
    &+ \underbrace{\mathbb{E}\left\{\tilde{Z}_m^*\right\}}_{=0} \mathbb{E}\left\{\left|\hat{Z}_n\right|e^{j\arg\left(\hat{Z}_m\right)}e^{-j\tilde{c}_n}e^{j\tilde{c}_m}\right\}\nonumber \\
    &+\underbrace{\mathbb{E}\left\{\tilde{Z}_n\right\}}_{=0}\underbrace{\mathbb{E}\left\{\tilde{Z}_m^*\right\}}_{=0}\mathbb{E}\left\{e^{-j\arg\left(\hat{Z}_n\right)}e^{j\arg\left(\hat{Z}_m\right)}e^{-j\tilde{c}_n}e^{j\tilde{c}_m}\right\}\nonumber \\
     \stackrel{(b)}{=}&\mathbb{E}\left\{ \left|\hat{Z}_n\right|\right\}\mathbb{E}\left\{\left|\hat{Z}_m\right|\right\}\mathbb{E}\left\{e^{-j\tilde{c}_n}e^{j\tilde{c}_m}\right\} \nonumber \\
    \stackrel{(c)}{=}&\frac{\pi}{4}\frac{(N+1)P_{\mathrm{pilot}}\left(\frac{\alpha\beta M}{N}\right)^2}{(N+1)P_{\mathrm{pilot}}\frac{\alpha\beta M}{N}+\sigma^2} \mathbb{E}\left\{e^{-j\tilde{c}_n}e^{j\tilde{c}_m}\right\},
\end{align}

where we used the independence of the channel estimation errors $\tilde{Z}_n$ and $\tilde{Z}_m$ from each other for $n\neq m$ and from $\hat{Z}_n$ and $\tilde{c}_n$ in $(a)$. In $(b)$, we have used that the random variables  $|\hat{Z}_n|$ and $|\hat{Z}_m|$ for $n\neq m$ and $\tilde{c}_n$ are mutually independent. In $(c)$, we have inserted \eqref{eq:exp-abs-Zhat}. Obtaining the closed-form expression for the last expectation above requires some more steps. 

By the law of iterated expectation, we can write the last expectation as
\begin{align}
   \mathbb{E}\left\{e^{-j\tilde{c}_n}e^{j\tilde{c}_m}\right\}= \mathbb{E}\left\{\mathbb{E}\left\{e^{-j\tilde{c}_n}e^{j\tilde{c}_m} \big | \arg\left(\hat{p}\right) \right\}\right\} \label{eq:independence},
\end{align}
where the innermost expectation on the right-hand side is over the random variables $\tilde{c}_n$, whereas the outermost expectation is over the random variable $\arg\left(\hat{p}\right)$. Since $\tilde{c}_n$ and $\tilde{c}_m$ for $n\neq m$ are functions of independent random variables $\arg(\hat{Z}_n)$ and $\arg(\hat{Z}_m)$  given $\arg(\hat{p})$,  $\tilde{c}_n$ and $\tilde{c}_m$ are independent given $\arg\left(\hat{p}\right)$.  Therefore, \eqref{eq:independence} becomes
\begin{align}
   \mathbb{E}\left\{e^{-j\tilde{c}_n}e^{j\tilde{c}_m}\right\}=& \mathbb{E}\left\{\mathbb{E}\left\{e^{-j\tilde{c}_n}\big | \arg\left(\hat{p}\right)\right\}\mathbb{E}\left\{e^{j\tilde{c}_m} \big | \arg\left(\hat{p}\right) \right\}\right\} \nonumber \\
   =& \mathbb{E}\left\{\mathbb{E}\left\{e^{-j\tilde{c}_n}\right\}\mathbb{E}\left\{e^{j\tilde{c}_m}\right\}\right\} \nonumber \\
   =&\mathbb{E}\{\sinc^2(1/K)\} =\sinc^2(1/K)
   \label{eq:independence2}
\end{align}
due to the fact that $\tilde{c}_n|\arg\left(\hat{p}\right)\sim \mathcal{U}[-\pi/K,\pi/K)$ for any $\arg\left(\hat{p}\right)$. We also used the moment generating function from \eqref{eq:MGF}. Arranging the terms, we obtain the closed-form expression of $\mathcal{B}$ as 
\begin{align}
    \mathcal{B} = \frac{\pi}{4N}\frac{(\alpha \beta M)^2P_{\mathrm{pilot}}/\sigma^2 \sinc^2(1/K)}{\alpha \beta M P_{\mathrm{pilot}}/\sigma^2+\frac{N}{N+1}}. \label{eq:mathcalB}
\end{align}

Inserting $\mathcal{A}$ from \eqref{eq:mathcalA} and $\mathcal{B}$ from \eqref{eq:mathcalB} into \eqref{eq:result-lemma2}, we finally obtain \eqref{eq:lemma2}, which concludes the proof.
\end{IEEEproof}


Theorem~\ref{th:general} reveals how imperfect CSI and quantization lower the SNR. Note that the case of perfect CSI and perfect phase-shift resolution in \eqref{eq:exactSNR} is a special case of the result in \eqref{eq:lemma2} that is obtained from Theorem~\ref{th:general} as $K\rightarrow \infty$ and $P_{\mathrm{pilot}} \rightarrow \infty$. This means that increasing $P_{\mathrm{pilot}}$ and $K$ will increase the SNR up to a certain upper bound. It can also be noted that increasing the number of subarrays $N$ is beneficial to the SNR. This behavior is illustrated in Fig.~\ref{fig:bigformula} for a 1-bit RIS with the UE transmitting at $P_{\mathrm{data}}=100$\,mW at a direct path loss of $\rho=-110$\,dB and varying values of $P_{\mathrm{pilot}}$. The resulting average SNR is compared to the perfect CSI and quantization-free case given by \eqref{eq:exactSNR}. We see that the quantization degraded the signal quality by up to 2.6\,dB for individually configured subarrays. A thing to note is that quantization effects are present in two terms in \eqref{eq:lemma2} through the $\sinc$ functions and affect both the cross term and the RIS term. The relevant simulation parameters, path loss, bandwidth, RIS elements, and quantization states are found in Table \ref{tab:table1}. The curves in the figure are computed using the analytical expression in Theorem~\ref{th:general}, while the overlapping stars are obtained from 1000 Monte Carlo trials, which confirms the validity of the formula.

\begin{table}[h!]
  \begin{center}
    \caption{Simulation Parameters}
    \label{tab:table1}
    \begin{tabular}{|l|r|} 
     \hline
      \textbf{Parameter} & \textbf{Value} \\
      \hline
      RIS to BS path loss: $\alpha$ & $-60$\,dB \\
      UE to RIS path loss: $\beta$ & $-80$\,dB \\
      Data SNR target $\gamma_d$ & $20$\,dB  \\
      Quantization states $K$ & $2$ \\
      RIS elements: $M$ & $1024$  \\
      Bandwidth: $B$ & $100$\,MHz  \\
      UE passive circuit power: $P_{\mathrm{circuit}}$ & $10$\,mW  \\
      Complex Gaussian thermal receiver noise: $\sigma^2$ & $-123.9$\,dB \\
       \hline
    \end{tabular}
  \end{center}
\end{table}

\begin{figure}[t!] \centering
\begin{overpic}[width=.99\columnwidth]{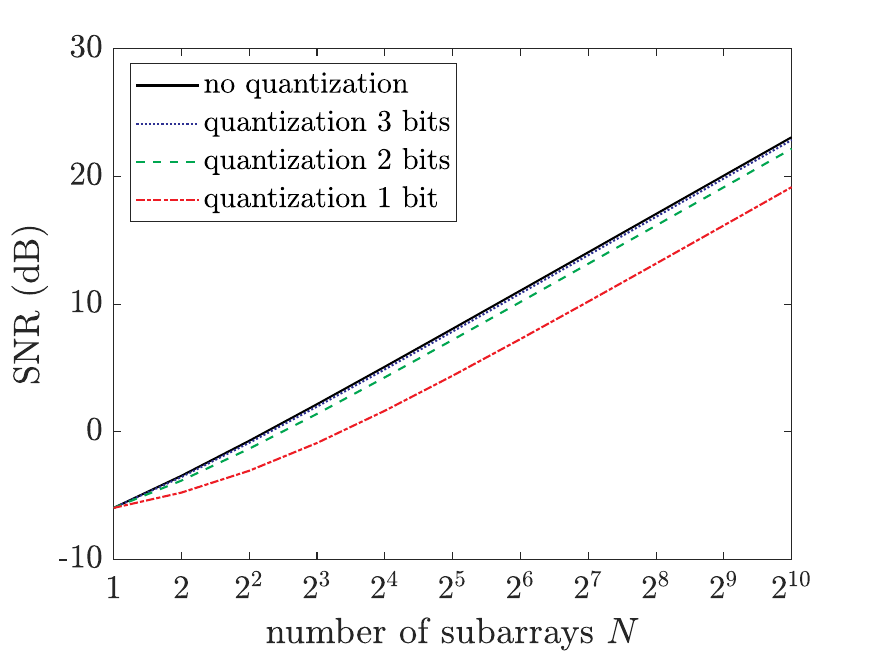} 
\put(91,56.5){\vector(0,-1){3}}
\put(91,55.5){\vector(0,1){3}}
\put(92,55.6){$3.9$ dB}
\end{overpic} 
\caption{Impact of quantization of order 1, 2, and 3 on the SNR versus $N$ for $\alpha=-80$\,dB, $\beta=-60$\,dB, and $\rho=0$ with $M=1024$ RIS elements being divided into subarray sizes of $N=2^R$ for $R\in \{ 0,1, \dotsc ,10\}$. The baseline of no quantization is also shown. Note the 3.9 dB performance gap due to quantization from Lemma~\ref{lem:upperbound-SNR-loss} in the top right corner.}  \label{fig:quantSNR}
\end{figure}

\begin{figure}[t!] \centering
\centerline{\includegraphics[trim=8 2 25 15,clip,width=0.99\columnwidth]{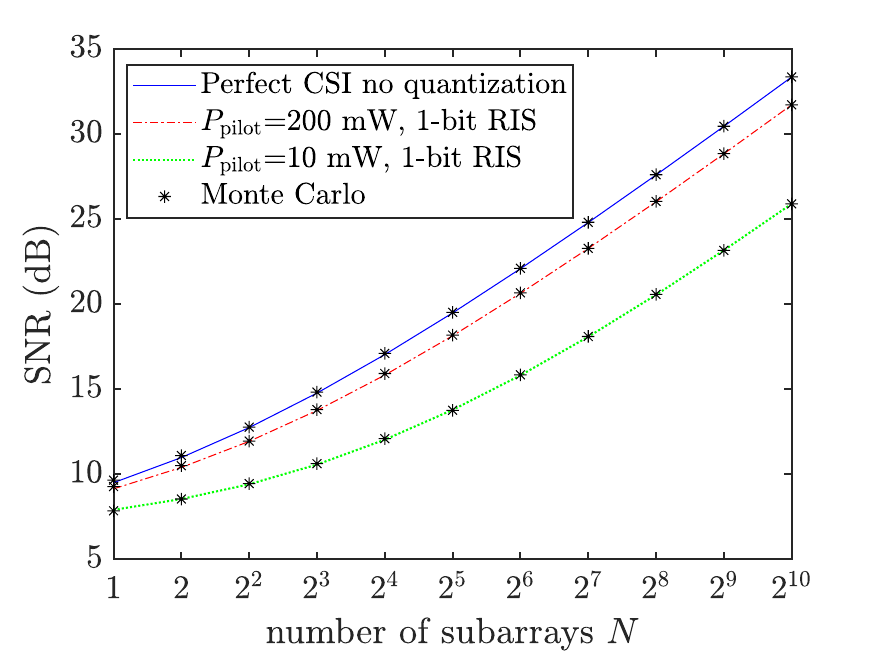}}  
\caption{The average SNR versus $N$ with $M=1024$ RIS elements being divided into $N=2^R$ subarrays for $R \in \{ 0,1, \dotsc ,10\}$. Different pilot powers for a 1-bit RIS are compared to the ideal perfect CSI and quantization-free case. {\color{black}The black stars are computed as averages from 10\,000 Monte Carlo simulations.}
} \label{fig:bigformula}
\end{figure}

\subsection{Special case: Perfect CSI with quantized RIS phase-shift resolution} 

To showcase the fundamental impact of having quantized phase-shift resolution, we will examine the SNR with perfect CSI and no direct path. The resulting SNR is obtained from Theorem~\ref{th:general} by inserting $\rho=0$ into \eqref{eq:lemma2} and letting $P_{\mathrm{pilot}} \rightarrow \infty$:
\begin{equation} \label{eq:SNR2}
\mathbb{E} \left\{\overline{\mathrm{SNR}}\right\}= \frac{P_\mathrm{data}}{\sigma^2}\alpha \beta M \left(1+ \frac{\pi}{4}(N-1) \sinc^2(1/K)\right).
\end{equation}
We observe that as the quantization resolution decreases, i.e., $K$ decreases, the average SNR reduces proportionally. On other hand, when $K\rightarrow \infty$, $\sinc^2(1/K) \rightarrow 1$ and the average SNR converges to \eqref{eq:exactSNR} with $\rho=0$, which corresponds to perfect CSI and perfect RIS phase-shift resolution. 

Using the above result, the following lemma provides an upper bound on the average SNR loss for the general case of imperfect CSI and finite phase-shift quantization resolution with or without a direct path.


\begin{lemma} \label{lem:upperbound-SNR-loss}
The relative average SNR loss due to $b$-bit phase-shift quantization is upper bounded by $\sinc^2(1/K)$, where $K=2^b$ denotes the number of equally-spaced quantization states on the unit circle.
\end{lemma}

\begin{IEEEproof}
We note that the average SNR in \eqref{eq:lemma2} is a monotonically increasing function of the pilot power $P_{\mathrm{pilot}}$. Without CSI, i.e., $P_{\mathrm{pilot}}=0$, quantization cannot degrade the performance.  The contributions to the average SNR come from the direct path and the RIS-assisted path. Quantization effects only degrade the performance of the RIS-assisted path, as can be seen by inserting $\alpha=\beta=0$ into \eqref{eq:lemma2}. In this case, the average SNR becomes independent of the quantization resolution $K=2^b$. The greatest degradation thus occurs if the RIS has perfect CSI and when the direct path contribution is zero, i.e., $\rho=0$. In this case, the average SNR is given by \eqref{eq:SNR2} and can be lower bounded as
\begin{align}
\mathbb{E} \left\{\overline{\mathrm{SNR}}\right\}=&\frac{P_\mathrm{data}}{\sigma^2}\alpha \beta M \left(1+ \frac{\pi}{4}(N-1) \sinc^2(1/K)\right) \nonumber \\
&>\frac{P_\mathrm{data}}{\sigma^2}\alpha \beta M \sinc^2(1/K)\left(1+ \frac{\pi}{4}(N-1) \right).
\end{align}
The lower bound equals the average SNR for the case of perfect CSI and  perfect phase-shift resolution with $\rho=0$ from \eqref{eq:exactSNR} multiplied by $\sinc^2(1/K)$, thus the relative SNR loss compared to the ideal case cannot be larger than $\sinc^2(1/K)$. 
\end{IEEEproof}

The impact of quantization and the upper bound from Lemma~\ref{lem:upperbound-SNR-loss} is illustrated in Fig.~\ref{fig:quantSNR}. Note that $\sinc^2(1/2)=-3.9$\,dB for the 1-bit quantized RIS, which represents the uttermost average SNR loss in logarithmic scale. The SNR loss is only $-0.9$\,dB when using a 2-bit RIS and $-0.2$\,dB with a 3-bit RIS, which explains why those performance gaps are barely visible in the figure. The performance loss due to the
$\sinc^2(1/K)$ function matches the results in \cite{BITSMEMS,wu2019beamforming} and is an RIS equivalent to the ``quantization lobe'' phenomena described in \cite{hansen2009phased} for phased arrays.

{\color{black}

\subsection{Multi-antenna setups}
The RIS-assisted single-input single-output (SISO) model considered in this paper can also be extended to system models involving multiple antennas at the BS. We will outline the extension in this section, for the case without a direct path, which we previously noted is the case when the RIS is particularly useful. We consider a BS equipped with $I$ antennas in a uniform linear array (ULA) serving a single-antenna UE. Because the BS and RIS are assumed to be deployed in LoS of each other, the received signal $\vect{y}=[y_1, \dots, y_I]^{\Ttran}$ at the $I$ antennas is given similarly to \eqref{eq:received-signal2} as
\begin{equation}
    \label{multi-extension}
    \vect{y}= 
    \vect{a}\left(\sum_{n=1}^{N} Z_{n} e^{-j c_n} \right)x+\vect{w},
\end{equation}
where 
$\vect{a}$ is the array response vector of a ULA \cite[Ch.~7]{massivemimobook}:
\begin{equation}
    \vect{a}=\begin{bmatrix} 1 \\
    e^{j2\pi\Delta \sin(\phi) \cos(\theta)} \\
    \vdots \\ 
    e^{j2\pi(I-1)\Delta \sin(\phi) \cos(\theta)}  \end{bmatrix}.
\end{equation}
The azimuth angle $\phi$ and elevation angle $\theta$ depend on the geometric relation between the BS and the RIS, $\Delta$ is the inter-antenna distance in the ULA in terms of wavelength, and $\vect{w}$ is the noise vector with i.i.d. $ \CN (0,\sigma^2)$ entries. 
The SNR of the received signal in \eqref{multi-extension} is maximized by 
maximum ratio (MR) combining, i.e., taking an inner product with $\frac{\vect{a}}{\|\vect{a}\|}$ \cite{massivemimobook}. The resulting processed received signal is 
\begin{equation}
    \label{multi-extension2}
    \frac{\vect{a}^{\Htran}}{\|\vect{a}\|}\vect{y}=\sqrt{I}\left(\sum_{n=1}^{N} Z_{n} e^{-j c_n} \right)x+\tilde{w},
\end{equation}
where $\tilde{w}=\frac{\vect{a}^{\Htran}}{\|\vect{a}\|}\vect{w}$ has the same distribution as each entry of $\vect{w}$. We notice that the magnitude of the term within parentheses is maximized by setting $c_n=\arg(Z_n)$. This results in the average SNR value 
\begin{align}
\label{SIMO-SNR}
&\mathbb{E}\left\{  \overline{\mathrm{SNR}}\right\} =I\frac{P_\mathrm{data}}{\sigma^2}\Bigg(\alpha \beta M \nonumber \\
&+ \frac{\pi}{4}\left(1-\frac{1}{N}\right)\frac{(\alpha \beta M)^2P_{\mathrm{pilot}}/\sigma^2 \sinc^2(1/K)}{\frac{\alpha \beta M}{N} P_{\mathrm{pilot}}/\sigma^2+\frac{1}{N+1}}\Bigg).
    \end{align}
This equation is the same as \eqref{eq:lemma2}, except for an added gain factor $I$ (the receive beamforming gain) and without a direct path since $\rho=0$. Therefore, whatever methods we develop to minimize the energy consumption for the model in \eqref{eq:lemma2} can easily be applied to the exemplified multi-antenna setups by including the scaling factor $I$. It is also worth mentioning that the same conclusion will be inferred after doing the analysis for the BS array geometries other than ULA. 

In more complex multi-antenna setups, there is no closed-form expression for the optimal RIS configuration. It is then impossible to compute the average SNR in closed form or to minimize the energy consumption analytically. However, the general methodology provided in the remainder of this paper can be easily followed to identify the preferred operation (e.g., power allocation and subarray sizes) numerically.
}

\section{Minimizing the Energy Consumption}

We consider the practical scenario where the UE wants to transmit a data payload with finite-sized information content {\color{black}over a narrowband channel with imperfect CSI.}
This payload data is mapped to $L$ complex modulation symbols, based on some joint modulation and coding scheme, which requires an average SNR of $\gamma_d$ for successful decoding (with some desirably low error probability). The value of $L$ greatly depends on how many bits the BS has scheduled for the UE together with the modulation and coding scheme. 
As a reference, in 5G OFDM-based systems, if a UE is scheduled on a single resource block during one time slot, it would transmit $12 \cdot 14 = 168$ symbols. If the UE is scheduled on the whole carrier, consisting of 273 resource blocks in the mid-band, during one time slot it would transmit $ 12 \cdot 273 \cdot 14 = 45\,864$ symbols. Depending on the channel coherence time, the pilot measurement may be valid for multiple time slots \cite{dahlman20205g}. This leads to the possibility of $L$ attaining values in the order of $L=200$ to $L=50\,000$ or higher. We will consider cases of small and large payloads herein, but treat both $L$ and $\gamma_d$ as given parameters in the problem formulation.

We want to minimize the energy consumption at the UE while ensuring successful decoding at the BS.
The energy consumption $E(N,P_\mathrm{pilot})$ of transmitting the $N+1$ pilots during channel estimation (with $N$ subarrays) and the $L$ data symbols is

\begin{equation} \label{eq:energy-consumption1}
    E(N,P_\mathrm{pilot})= (N+1)\frac{P_\mathrm{pilot}}{B}+L \frac{P_\mathrm{data}}{B}+(L+N+1)\frac{P_{\mathrm{circuit}}}{B},
\end{equation}
where $B$ is the symbol rate measured in symbol/s.

To optimize the battery life on the UE side, we should also include the power dissipation required to keep the device in active mode. After all, a longer transmission time will consume more power. We therefore included a term proportional to $P_{\mathrm{circuit}}$ in \eqref{eq:energy-consumption1}, which represents the energy consumption of the analog circuits. Note that we have omitted the energy consumption due to digital computations since all the complex operations are made at the BS side.

Since an average SNR of $\gamma_d$ is needed for the data transmission, we obtain from  \eqref{eq:lemma2} the required transmission power as shown in \eqref{eq:Pdata} at the top of the next page.
\begin{figure*}
\begin{equation} 
P_\mathrm{data}= \frac{\sigma^2 \gamma_d}{\rho+\alpha \beta M + \frac{\pi}{2}\frac{\rho \alpha \beta M P_{\mathrm{pilot}}/\sigma^2 \sinc(1/K)}{\sqrt{\left(\rho P_{\mathrm{pilot}}/\sigma^2+\frac{1}{N+1}\right)\left(\frac{\alpha \beta M}{N} P_{\mathrm{pilot}}/\sigma^2+ \frac{1}{N+1}\right)}}
+\frac{\pi}{4}\left(1-\frac{1}{N}\right)\frac{(\alpha \beta M)^2P_{\mathrm{pilot}}/\sigma^2 \sinc^2(1/K)}{\frac{\alpha \beta M}{N} P_{\mathrm{pilot}}/\sigma^2+\frac{1}{N+1}}
}. \label{eq:Pdata}
\end{equation}
\hrulefill
\end{figure*}

\subsection{Problem formulation}

To determine the subarray size $N$ along with the pilot power $P_\mathrm{pilot}$ that minimize the energy consumption, we solve the following optimization problem:

\begin{equation} \label{eq:optimization}
\begin{aligned}
 \underset{N, P_\mathrm{pilot}}{\textrm{minimize}} \quad & E(N,  P_\mathrm{pilot})\\
\textrm{subject to} \quad & N\in \{1,\ldots,M\}, \quad \frac{M}{N} \in \mathbb{Z}, \quad P_\mathrm{pilot}\geq 0  .
\end{aligned}
\end{equation}
 The optimization variables are the subarray size $N$ and the pilot transmit power $P_\mathrm{pilot}$. It is only when $1\leq N^{\star}\leq M$, however, that the solution $N^{\star}$ is viable. When all subarrays consist of the same number of elements then $M/N$ must also be an integer. {\color{black}The optimization is based on the long-term channel statistics (i.e., $\alpha, \beta, \rho$) since the exact channel realizations are estimated only after the pilots have been transmitted.} The total number of RIS elements, $M$, the number of data symbols $L$, and the average SNR requirement for the data transmission, $\gamma_d$, are given parameters. The required data power $P_{\mathrm{data}}$ is determined in terms of optimization variables and fixed parameters according to \eqref{eq:Pdata}.

\subsection{Numerical optimization}
To optimize the energy consumption for a given number of information symbols, one needs to find the solution $(N^{\star},P_{\mathrm{pilot}}^{\star})$ to \eqref{eq:optimization} that minimizes the energy consumption. The simulation parameters are given for two realistic different scenarios. In the first scenario, the UE has a small payload to transmit and a weak direct channel to the BS. In the second scenario, the UE has a large payload and a strong direct channel to the BS. 
The energy consumption for these two scenarios is illustrated Fig.~\ref{fig:3D}(a) and Fig.~\ref{fig:3D}(b), respectively, as a function of $N$ and $P_\mathrm{pilot}$. The corresponding simulation parameters are given in Table~\ref{tab:table1}. The surfaces in both figures are smooth for a wide variety of input parameters and thus we may utilize an alternating optimization algorithm, which optimizes one variable at the time, similar to the one utilized in \cite{bjornson2015optimal} to find the optimal solution. In Fig.~\ref{fig:3D} we employed a gradient steepest descent algorithm, which due to the smoothness of the function converged to the global optimum within a few iterations. The results are shown in black for 5, 10 and 15 iterations.

\begin{figure*}
\centering
\begin{subfigure}{.49\textwidth}
  \centering
  \includegraphics[width=.9\linewidth]{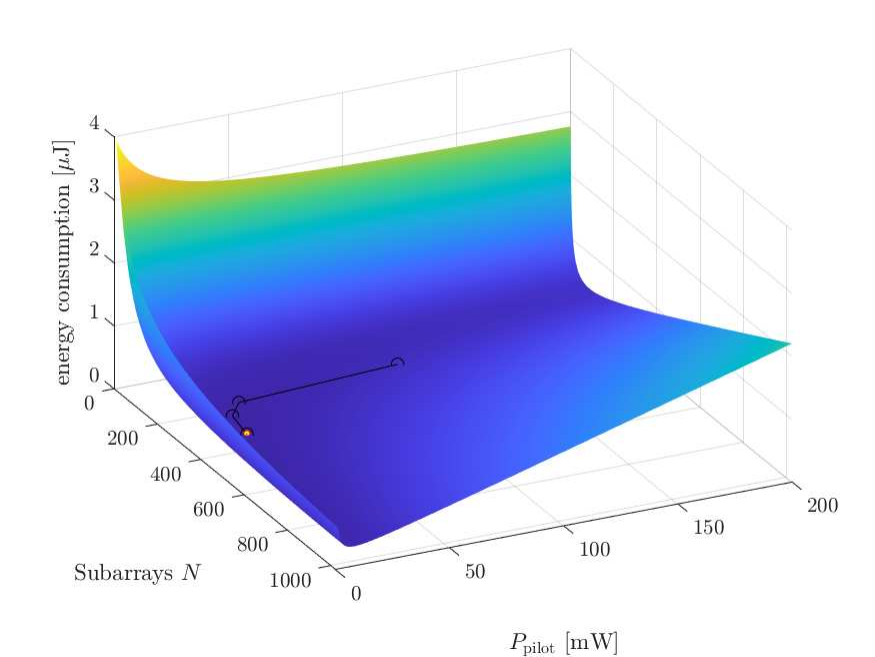}
  \caption{}
  \label{fig:m110L200}
\end{subfigure}%
\begin{subfigure}{.49\textwidth}
  \centering
  \includegraphics[width=.9\linewidth]{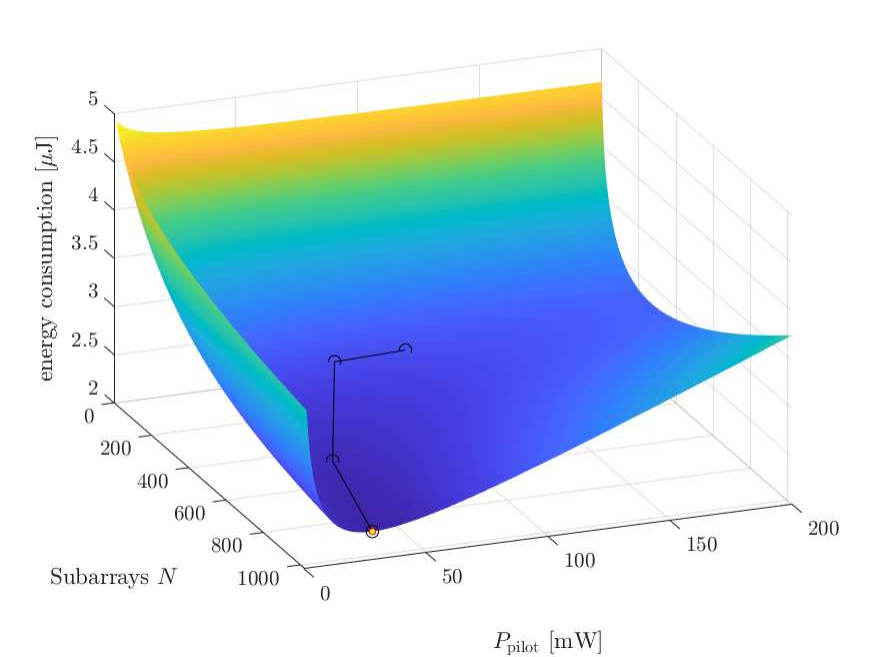}
  \caption{}
  \label{fig:m90L10000}
\end{subfigure}
\caption{Energy consumption $E(N,P_{\mathrm{pilot}})$ {\color{black} in a scenario with imperfect CSI} when transmitting $L=200$ symbols and the direct path is weak ($\rho=-110$\,dB) in (a) and when transmitting $L=10\,000$ symbols and the direct path is strong ($\rho=-90$\,dB) in (b). The results  obtained by the numerical optimization based on steepest descent are shown in black for 5, 10, and 15 iterations.}
\label{fig:3D}
\end{figure*}

As can be seen, Fig.~\ref{fig:3D}(a) shows that there is a global optimal solution at $N^{\star}=403$ and $P_{\mathrm{pilot}}^{\star}=19\,\mathrm{mW}$ when transmitting 200 symbols and the direct path is weak. However, $N^{\star}$ will have to be rounded to be a power of 2. In this case, 256 is a better solution than 512. Once the optimal subarray size has been determined, the pilot power should consequently also be fine-tuned through a gradient-based method to bring the energy consumption down, in this case to $20\,\mathrm{mW}$. For the scenario in Fig.~\ref{fig:3D}(b) when transmitting 10\,000 symbols and the direct path is strong, $N^{\star}=1024$ and $P_{\mathrm{pilot}}^{\star}=28\,\mathrm{mW}$ are obtained.


\subsection{Payload length}

Since practical system will transmit payloads of variable length $L$ depending on the UE's traffic, there is no one-size-fits-all solution but different solutions to \eqref{eq:optimization} are needed by different UEs and at different times.
Hence, it is important to determine how the optimal energy consumption solution parameters $(N^\star, P_\mathrm{pilot}^\star)$ behave together. We have already established that increasing either will also increase the received SNR at the cost of higher energy consumption. For a varying payload length, the energy consumption optimization problem \eqref{eq:optimization} was solved. The optimal solutions are shown in Figs.~\ref{fig:NiceImage}(a) and \ref{fig:NiceImage}(c). Also shown is the required data power $P_\mathrm{data}$ in Fig.~\ref{fig:NiceImage}(b) and pilot energy in Fig.~\ref{fig:NiceImage}(d).

\begin{figure*}
\centering
\begin{subfigure}[b]{0.475\textwidth}
\centering
   \includegraphics[width=.9\textwidth]{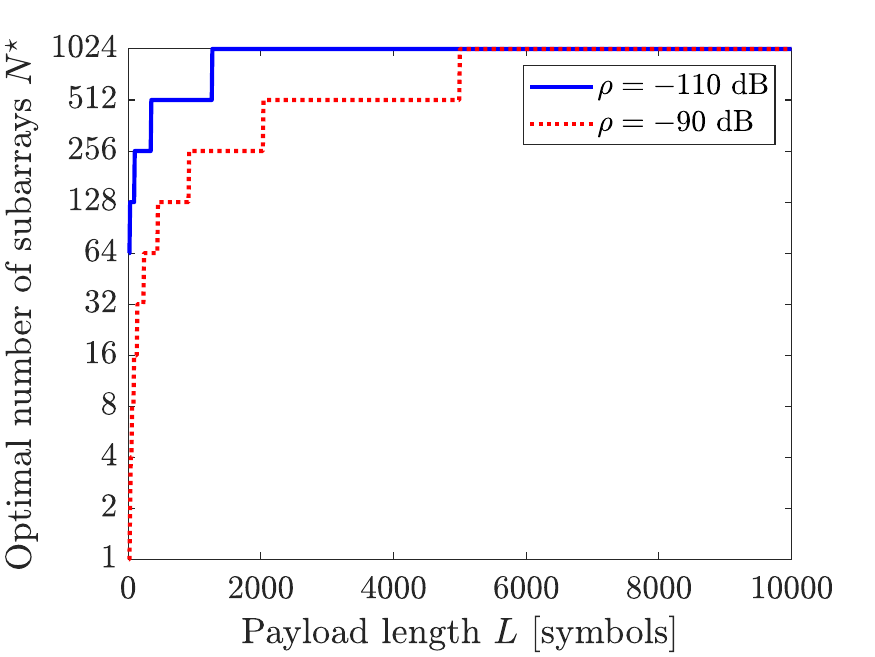}
   \caption{}
   \label{fig:LvsN}
\end{subfigure}
\hfill
\begin{subfigure}[b]{0.475\textwidth}
\centering
   \includegraphics[width=.9\textwidth]{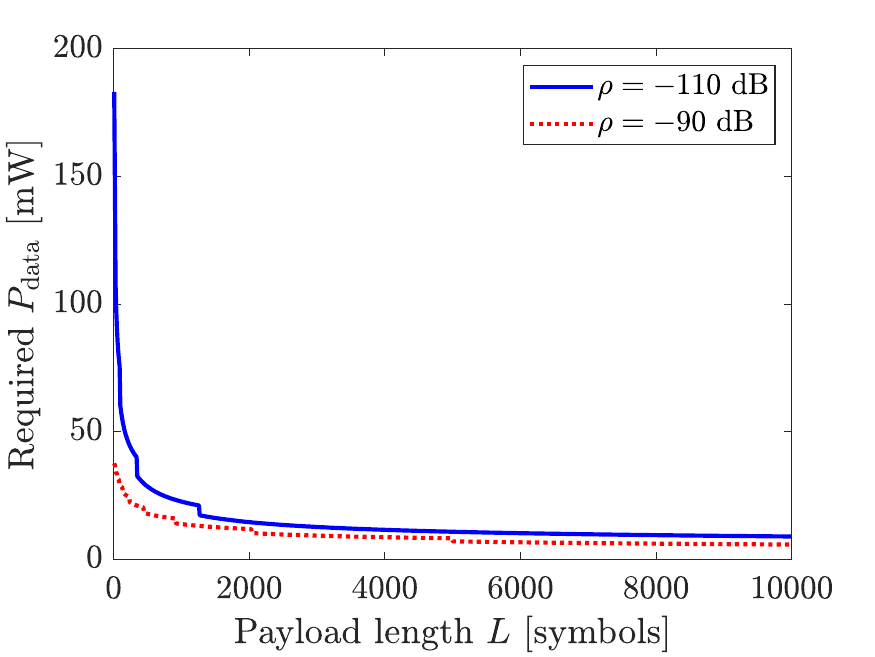}
   \caption{}
   \label{fig:LvsD}
\end{subfigure}
\hfill
\begin{subfigure}[b]{0.475\textwidth}
\centering
   \includegraphics[width=.9\textwidth]{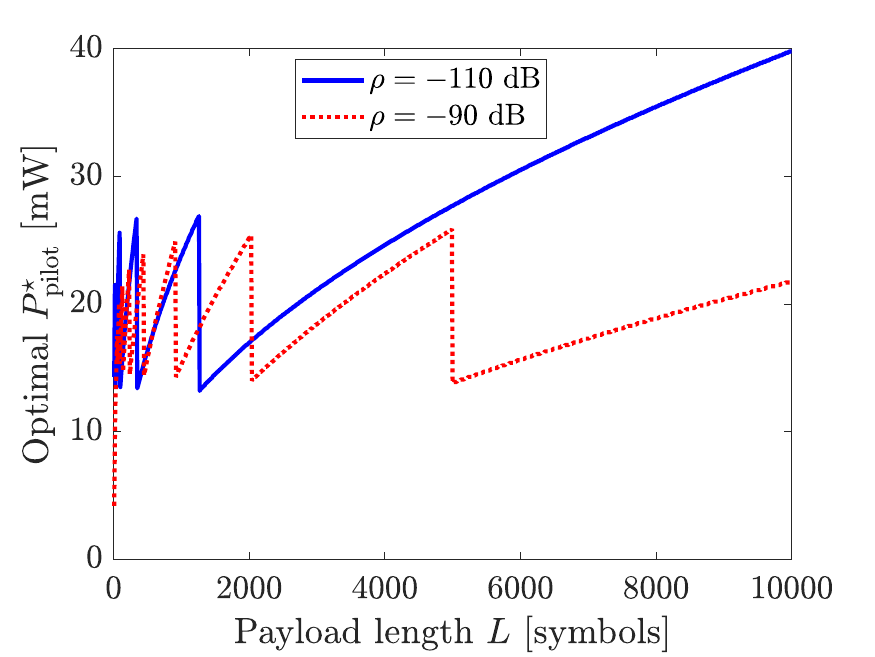}
   \caption{}
   \label{fig:LvsP}
\end{subfigure}
\hfill
\begin{subfigure}[b]{0.475\textwidth}
\centering
   \includegraphics[width=.9\textwidth]{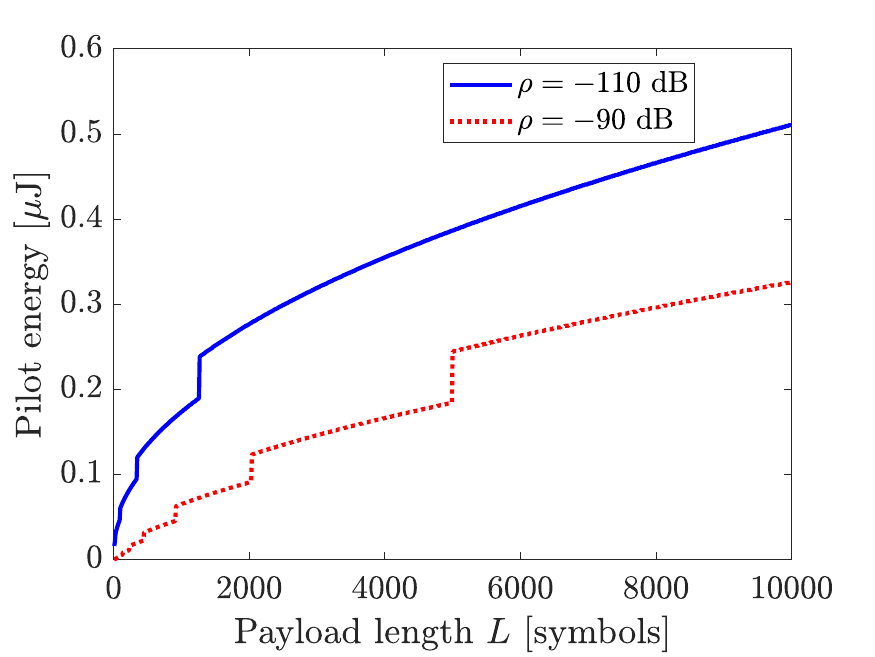}
   \caption{}
   \label{fig:LvsE}
\end{subfigure}
\hfill
\begin{subfigure}[b]{0.475\textwidth}
\centering
   \includegraphics[width=.9\textwidth]{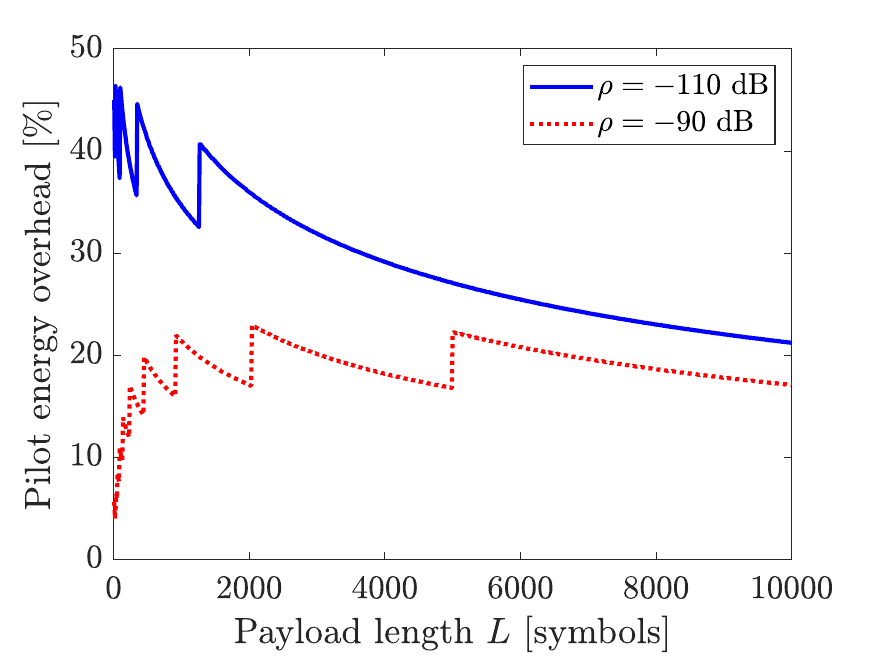}
   \caption{}
   \label{fig:LvsOH}
\end{subfigure}

\caption[Three figures]{The optimal energy consumption minimizing subarray size (a) and pilot power (c), required data power (b), total pilot energy (d), {\color{black}and the percentage of energy spent on the pilot signaling (e)}, for different payload lengths are shown for a weak ($\rho=-110$\,dB) and strong ($\rho=-90$\,dB) direct path between the UE and BS.}
\label{fig:NiceImage}
\end{figure*}

Fig.~\ref{fig:NiceImage}(a) is a step-function, which shows that as the payload grows, it becomes more efficient to use more subarrays (which consequently are smaller). The red curve being under the blue indicates that the number of subarrays is higher in the case when the direct path loss $\rho$ is high, which means that the RIS is more useful in that case. The reason for this is that, the larger the payload is, the more we can benefit from reducing the data power by configuring the RIS on a per-element basis. By contrast, clustering the elements into subarrays is an effective strategy to bring the energy consumption down when the payload is sufficiently small (depending on the relative path loss as can be seen in Fig.~\ref{fig:NiceImage}(a)).

Fig.~\ref{fig:NiceImage}(c) displays a saw-tooth behaviour. {\color{black}As the payload grows larger, it becomes beneficial to increase the pilot power to get a more accurate channel estimation which mitigates the effects of imperfect CSI, and gives a more accurate RIS configuration}. The sudden drops in Fig.~\ref{fig:NiceImage}(c) coincide with increases in the number of subarrays in Fig.~\ref{fig:NiceImage}(a). As the payload $L$ grows when passing through the values corresponding to these jumps, it becomes better to reduce the pilot power but increase the number of subarrays. 

As $(N^\star, P_\mathrm{pilot}^\star)$ vary together, it is also of interest to examine how the total energy spent on pilot signals behaves. This is given by the terms $(N+1)\frac{P_\mathrm{pilot}+P_\mathrm{circuit}}{B}$ in \eqref{eq:energy-consumption1} and shown in Fig.~\ref{fig:NiceImage}(d). 
In general, the energy is steadily increasing as the payload grows. However, there are rapid increases that occur when transmitting with more subarrays. This is because longer pilot signals result in an increased circuit energy consumption.
{\color{black}Fig.~\ref{fig:NiceImage}(e) reveals that the percentage of energy spent on pilot signaling makes a substantial contribution to the overall energy consumption which should not be ignored, irrespective if subarrays are used or not. It also shows that if the direct path is strong ($\rho = -90$\,dB), it is best to spend relatively little energy on configuring the RIS for small payloads. But if the direct path is weak ($\rho = -110$\,dB), the overhead can be in the order of $40\%$. For larger payloads in the order of $10\,000$ symbols it makes up around $20\%$ of the total energy consumption, but steadily decreases as the payloads grow larger.}

We have observed that as the payload increases, the number of subarrays and/or pilot power increase.
In Fig.~\ref{fig:NiceImage}(b), we show the impact on the transmit power for data transmission. We notice that for larger payloads, the SNR target $\gamma_d$ can be reached with a lower transmit power $P_\mathrm{data}$ from the UE. This follows from the dependence of $P_\mathrm{data}$  on $P_\mathrm{pilot}$ and $N$ as given by \eqref{eq:Pdata}.

The best strategy for minimizing the energy consumption in scenarios with larger payloads (over $1\,000$ symbols long in Fig.~\ref{fig:NiceImage}(a) when the direct path is weak or $5\,000$ symbols when the direct path is strong) is to configure the RIS elements individually and set the pilot power to an appropriate level. The transmitted pilot power levels from $10$ to $40$\,mW in Fig.~\ref{fig:NiceImage}(c) can be translated to pilot SNR levels of approximately between $-5$ to $0$\,dB per subarray using the second line of \eqref{eq:SNR_pilot}. This means that a low-resolution estimation of the channel to each individual RIS element is generally preferable to a high-resolution estimation of a larger subarray from an energy consumption perspective.

\subsection{Special case study: Perfect CSI and phase-shift resolution} 

The general optimization problem \eqref{eq:optimization} is too complex for a precise mathematical analysis, but we can obtain analytical insights by considering the energy consumption with perfect CSI together with infinite RIS phase-shift resolution.
We can then also make use of the lower bound since its behavior is similar to the exact average SNR (as shown in Fig.~\ref{fig:ESNR}) but the expression is simpler. For notational convenience,  we also neglect the circuit power. By perfect CSI, we mean that the RIS can be configured so that the losses due to the suboptimal phase-shift configuration of the subarrays are negligible.

Hence, we return to the tight lower SNR bound  $\overline{\mathrm{SNR}}_{\rm av,low}$ in (\ref{eq:Jensen}). We may assume that a pilot SNR of $\gamma_p$ defined as the second line in \eqref{eq:SNR_pilot} is needed for perfect channel estimation. Earlier simulations have shown that values of $20$\,dB or higher give negligible losses \cite{enqvist2022optimizing}. If we insert $\gamma_p$ into \eqref{eq:energy-consumption1}, the energy consumption becomes
\begin{equation}
    E(N)=\frac{\sigma^2 \gamma_p N}{B \alpha \beta M}
    +\frac{\sigma^2 \gamma_d L}{B}\frac{1}{\frac{\pi}{4} \left(\sqrt{\rho}+\sqrt{\alpha \beta M N}  \right)^2},
\end{equation}
which is no longer a function of $P_\mathrm{pilot}$ and whose first-order derivative with respect to $N$ is
\begin{equation} \label{eq:first-derivative}
    E'(N)=\frac{\sigma^2 \gamma_p}{B\alpha \beta M} 
    -\frac{\sigma^2 \gamma_d L}{B}
    \frac{\sqrt{\alpha \beta M}}{\frac{\pi}{4}\left(\sqrt{\rho}+\sqrt{\alpha \beta M N} \right)^3\sqrt{N}},
\end{equation}
while the second-order derivative is
\begin{equation}
    E''(N)=\frac{\sigma^2 \gamma_d L}{B}\frac{\sqrt{\alpha \beta M}(\sqrt{\rho}+4\sqrt{\alpha \beta M N})}{\frac{\pi}{2} (\sqrt{\rho}+\sqrt{\alpha \beta M N})^4N^{3/2}}.
\end{equation}
We first note that $E''(N)>0$ for every $N>0$ meaning that the function is convex. The same can be shown for $E(N)$ when the exact SNR expression is used instead of the bound. 
Consequently, a unique positive real-valued solution $N^{\star}$ to the equation $E'(N^{\star})=0$ that minimizes $E(N)$ exists.

The following lemma formalizes this result and proposes a way to find the optimal integer $N$ to the problem in \eqref{eq:optimization} in the special case considered in this section.

\begin{lemma} \label{lem2}
The optimal real-valued solution $N^{\star}$ that minimizes $E(N)$ exists and is obtained by finding the unique positive real-valued root of $E^{\prime}(N)=0$ via a bisection search. When $N$ is restricted to be a positive integer such that also $M/N$ is an integer, the optimal $N^{\star}$ can be found by evaluating $E(N)$ for the closest smaller and larger feasible values to the root and selecting the one giving the minimum $E(N)$.
 \end{lemma}
\begin{proof}
The real-valued $N$ that minimizes $E(N)$ can be found by equating the derivative of $E(N)$ to zero, i.e., $E'(N)=0$. This corresponds to equating the first term of $E'(N)$ in \eqref{eq:first-derivative} to the second term, which is decreasing with $N$. Moreover, the third term approaches infinity as $N \to 0$ and approaches zero as $N \to \infty$. Hence, there is a unique $N>0$ that satisfies $E'(N)=0$ which is guaranteed to be found by a bisection search. 
Moreover, since $E(N)$ is a convex function, the optimal $N$ that is both integer and divides $M$ can be found by checking the objective function $E(N)$ for the closest smaller and larger $N$ values that are feasible.
Note that these claims are valid also for the exact expression of $E(N)$. 
\end{proof}

As a general rule, once the UE transmits data, the configured RIS improves the SNR in the data transmission but the relative improvement depends on the strength of the different channel paths.

In Fig.~\ref{fig:E_weak} we again consider a scenario in which the UE has a strong direct path and one where the direct path is weak. The cascaded path between the BS and the UE via the RIS is the same in both cases. We analyze $E(N)$ for varying payload sizes as well. Fig.~\ref{fig:E_weak} shows that the optimal number of subarrays $N^{\star}$ that minimizes the energy consumption $E(N)$ increases as the number of data symbols $L$ increases. $N^{\star}$ is also increased when the direct channel gain $\rho$ becomes smaller.


The minimum to $E(N)$ is found either at $N=1$, $N=M$, or an interior point $1 < N < M$. The existence of an interior point solution can be determined using the first derivative in \eqref{eq:first-derivative}. There are three distinct cases for a convex function:

\begin{enumerate}
\item If $E'(1)<0$ and $E'(M)>0,$ the viable optimum is an interior point. This point can be quickly found utilizing the bisection search method \cite{burden2015numerical}.
\item If $E'(1)<0$ and $E'(M)<0$. The optimal solution is individually controlled RIS elements, i.e., $N=M$.
\item If $E'(1)>0$ and $E'(M)>0$. The optimal solution is to bundle all RIS elements into a single subarray, i.e., $N=1$.
\end{enumerate}

We can further analyze when these cases occur.
The condition $E'(1)<0$ for the lower bound is equivalent to
\begin{equation}
    \sqrt{\frac{\rho}{\alpha \beta M}}< \left( \left(\frac{4\gamma_d L }{\pi \gamma_p}\right)^{1/3}-1
    \right).
    \label{eq:ineq1}
\end{equation}

The (exact) energy consumption $E(N)$ 
is seen in Fig.~\ref{fig:E_weak}. For a strong direct path $\rho=-90$\,dB and small payload of $L=200$ symbols, the inequality (\ref{eq:ineq1}) does not hold.

\begin{figure}[t!]
\centerline{\includegraphics[trim=8 2 25 20,clip,width=0.99\columnwidth]{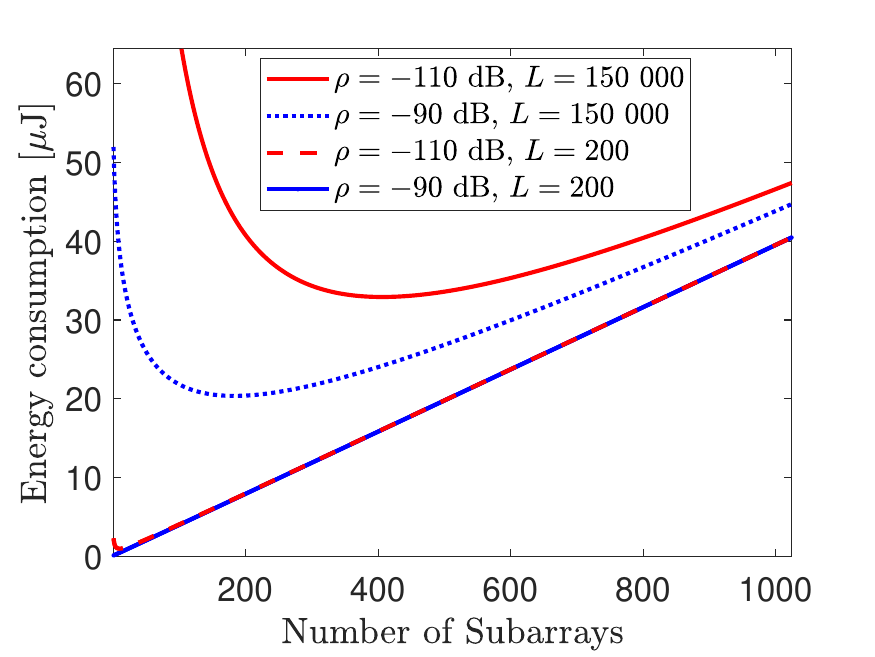}} 
\caption{Energy consumption $E(N)$ when transmitting $L$ symbols with varying gain of the direct channel $\rho$.}  \label{fig:E_weak}
\end{figure}

The condition $E'(M)>0$ for the lower bound can be shown to be equivalent to
\begin{equation}
   \sqrt{\frac{\rho}{\alpha \beta M}}> \left( \left(\frac{4\gamma_d L}{\pi\gamma_p\sqrt{M} }\right)^{1/3}-\sqrt{M}
    \right).
    \label{eq:ineq2}
\end{equation}

The term $\sqrt{\frac{\rho}{\alpha \beta M}}$ represents the average quotient between the direct path and the reflected paths' signal strength. If it is large, then (\ref{eq:ineq1}) will not hold because the direct path will be efficient enough so that the RIS elements will not aid the communication sufficiently to warrant the increased energy cost of transmitting the RIS pilots, leading to employing only a single large subarray as the most energy-efficient choice.

In contrast, if $\sqrt{\frac{\rho}{\alpha \beta M}}$ is very small, the direct path has a negligible contribution to the SNR, even if the RIS is configured randomly. If the inequality in \eqref{eq:ineq2} does not hold, it means that the system is starved for received signal power to the point that the most energy-efficient setup uses all the RIS elements steered individually. Note that the right-hand side of (\ref{eq:ineq2}) can attain negative values and that these inequalities were derived using the derivative in (\ref{eq:first-derivative}) for the lower bound expression. 

{\color{black}
\begin{remark}[Wideband Scenarios]
 In theory, when considering channel models for wideband scenarios, one has to build new optimization algorithms to optimize the data and pilot power dynamically across the different subbands along with the phase-shifts and subarray size of the RIS. In practice, however, the size of the subbands utilized would also depend on the payload length. For smaller payloads, like the ones considered in this paper, the narrowband approximation is an accurate representation. As an example, in current 5G systems narrowband channel reporting is used if the number of physical resource blocks (PRBs) is below $24$ \cite{3gpp}. This corresponds to up to $3864$ symbols per time slot and a transmission could take place across multiple time slots. Since the channel statistics are the same in all subbands, the subarray size that is optimal for one subband is likely jointly optimal for all subbands.
\end{remark}
}

\section{Conclusion}
Our study showed that the energy spent on transmitting pilot signals is an important factor when minimizing the energy consumption for transmitting a small-sized data payload over an RIS-assisted narrowband channel. However, rather than turning off RIS elements to reduce the dimensionality, we have shown that it is preferable to group the RIS elements into subarrays with a common reflection coefficient to maintain the aperture gain of the RIS (i.e., how much signal energy it collects), even if the beamforming gain is reduced by the grouping.
By optimizing the size of the subarrays, we can greatly reduce the energy consumption compared to the baseline where each RIS element is configured individually. 
In cases with a small payload and a relatively weak path to the RIS, it is beneficial to have many elements per subarray to reduce the energy spent on transmitting pilots, while the opposite was true in cases with large payloads and strong paths to the RIS.
If the RIS switches between aiding different UEs, the optimal subarray size will also change between UEs depending on their channel conditions and payload sizes.

Imperfect CSI due to low pilot power leads to errors in the phase-shift configurations, which in turn requires an increased data signal power to reach the SNR target. The effects of imperfect CSI are especially prevalent in setups involving a low number of subarrays while it has less impact on the case where the phase-shifts are optimized per element. The subarray size and pilot power are coupled together and depending on the geometry and size of the payload being transmitted, the situation calls for different configurations. As a general rule, in situations where the RIS was more useful (e.g., if the direct path between the UE and BS was weak or the payload was large) 
a high pilot power was beneficial. Another conclusion is that imperfect CSI heavily influences the energy that is consumed in total and analyses that are assuming perfect CSI (as our preliminary work \cite{enqvist2022optimizing}) are applicable mostly when the payload is much larger than the pilot. However, due to the large number of elements utilized in RISs, this may not always be the case in practice. Hence, attaining perfect CSI implicitly assumes overspending on pilot power, which can be a burden on energy consumption. {\color{black}Phase-shift quantization effects, which degrade the RIS performance, are more prevalent in situations with a weak gain between the UE and BS.}

The main takeaway message is that realistic payloads are often quite small and an RIS should be managed differently in practical scenarios, compared to what can be inferred from the capacity analysis in previous works, where the pilot overhead is negligible due to the infinite-sized payload. Since the subarray approach can provide decisive energy consumption reductions, it will be interesting to apply it to more advanced setups that incorporate multiple antennas at the BS and UE, as well as multiple RISs. {\color{black}Additionally, minimizing metrics involving energy consumption at the RIS and BS could also be considered.}

\bibliographystyle{IEEEtran}

\bibliography{Referenser}

\begin{IEEEbiography}[{\includegraphics[width=1in,height=1.25in,clip,keepaspectratio]{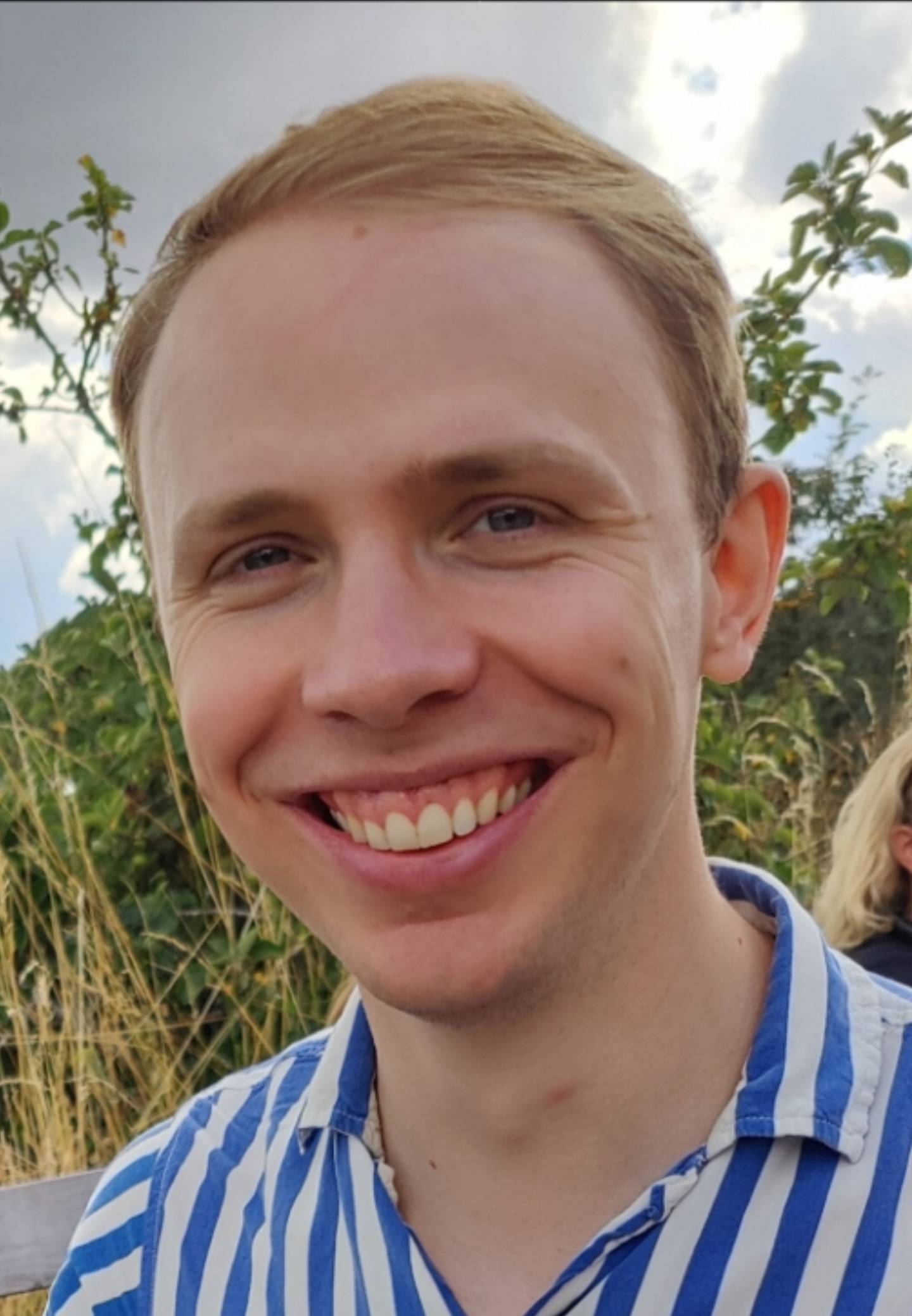}}]
    {Anders Enqvist} received the M.Sc. degree in electrical engineering from KTH Royal Institute of Technology, Stockholm, Sweden in 2021. Currently, he is a Ph.D. student in the Communication Systems department at the same university. His research focuses on energy efficiency, reconfigurable intelligent surfaces, and massive MIMO in wireless communication systems.

\end{IEEEbiography}
\begin{IEEEbiography}[{\includegraphics[width=1in,height=1.25in,clip,keepaspectratio]{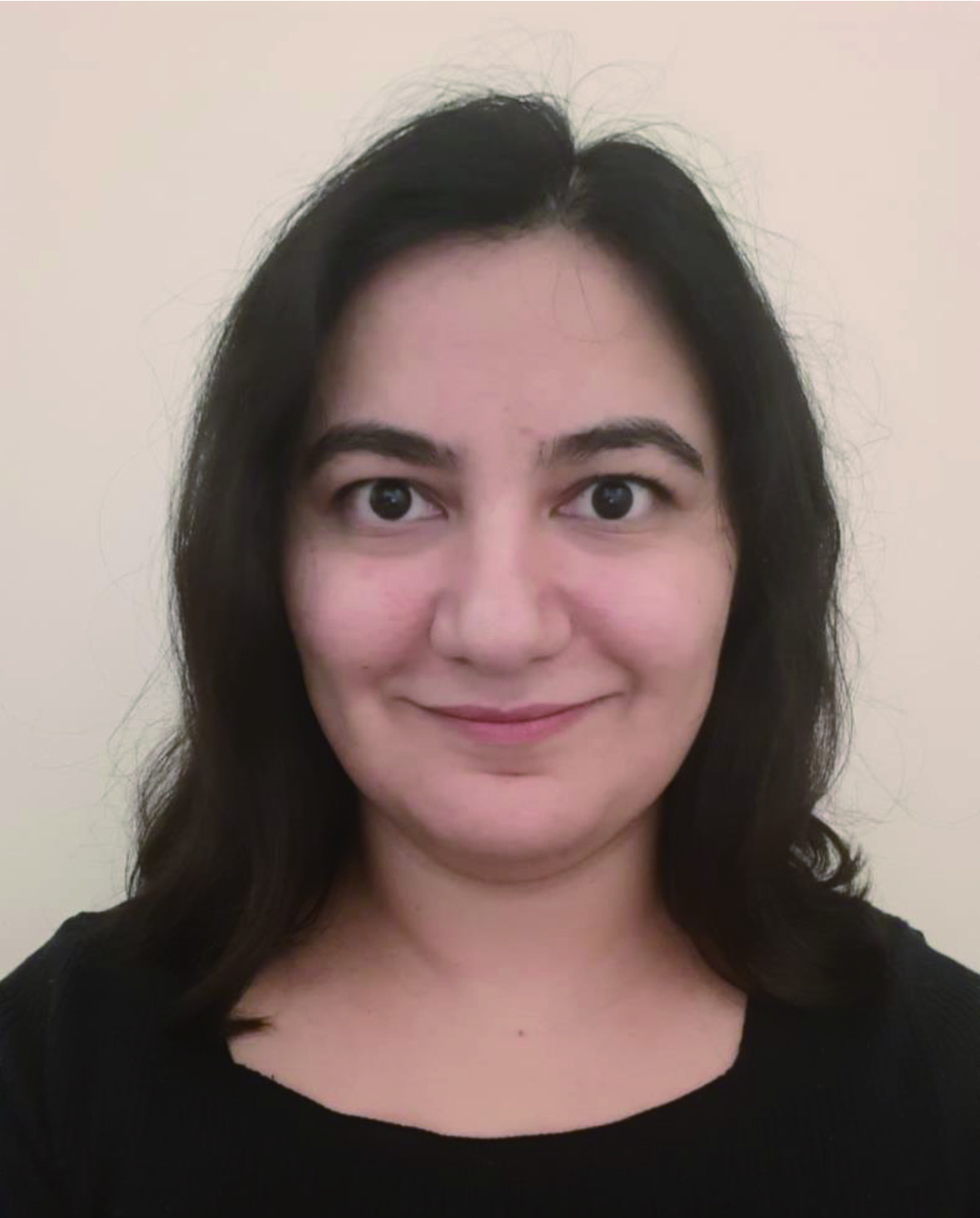}}]
    {{\"O}zlem Tu{\u{g}}fe Demir} received the B.S., M.S., and Ph.D. degrees in Electrical and Electronics Engineering from Middle East Technical University, Ankara, Turkey, in 2012, 2014, and 2018, respectively. She was a Postdoctoral Researcher at Link\"oping University, Sweden in 2019-2020 and at KTH Royal Institute of Technology, Sweden in 2021-2022. She is currently an Assistant Professor with the Department of Electrical and Electronics Engineering, TOBB University of Economics and Technology, Ankara, Turkey. She has authored the textbook \emph{Foundations of User-Centric Cell-Free Massive MIMO} (2021). Her research interests focus on signal processing and optimization in wireless communications, massive MIMO, cell-free massive MIMO, beyond 5G multiple antenna technologies, reconfigurable intelligent surfaces, machine learning for communications, mobile data analysis, and green mobile networks.
\end{IEEEbiography}

\begin{IEEEbiography}[{\includegraphics[width=1in,height=1.25in,clip,keepaspectratio]{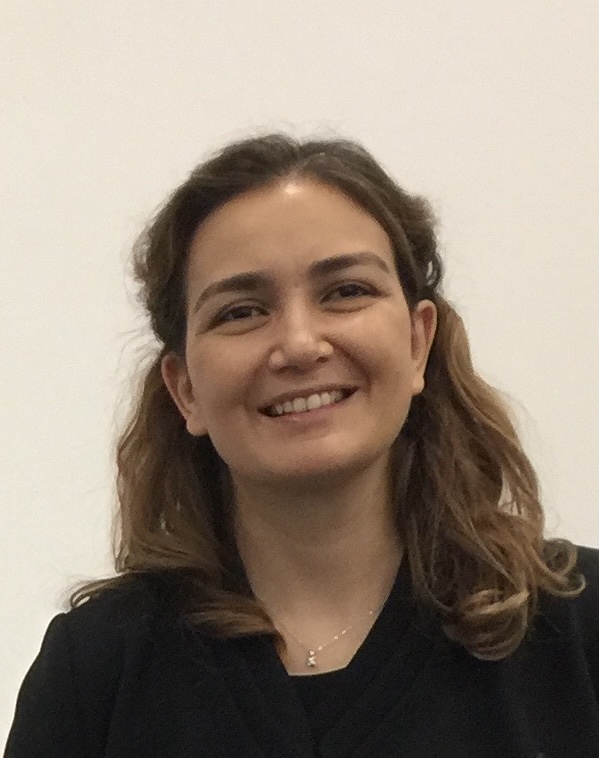}}]
    {Cicek Cavdar} is an Associate Professor, leading a research group on Intelligent Network Systems at the School of EECS at KTH Royal Institute of Technology in Sweden. She finished her Ph.D studies in Computer Science, University of California, Davis and in Istanbul Technical University, Turkey in 2009. After her PhD, she worked as an Assistant Professor in Computer Engineering Department, Istanbul Technical University. She served as symposium chair of IEEE ICC 2017 (GCSN),  held in Paris. She has chaired several workshops on green mobile broadband technologies and Green Mobile Networks last ten years co-located with IEEE ICC and Globecom.  She had leading roles in several European projects in the areas of aerial communications and green networks such as EU EIT Digital projects "5GrEEn: Towards Green 5G Mobile Networks" and "Seamless Direct Air to Ground Communications (DA2GC) in Europe- ICARO-EU". 2015-2018, she served as the leader of Swedish cluster for the EU Celtic Plus project SooGREEN "Service Oriented Optimization of Green Mobile Networks“, and from 2018 she has the same role in the EU Celtic Plus project AI4Green “Artificial Intelligence for Green Mobile Networks”. Her research interests include design and analysis of telecommunication networks with focus on beyond 5G mobile networks,  edge/cloud computing, big data in the network, URLLC, energy efficiency, and AI assisted mobile networks. She is the technical coordinator of EU Celtic-Next 6G-SKY - 6G for Connected Sky project which has started recently. 
\end{IEEEbiography}

\begin{IEEEbiography}[{\includegraphics[width=1in,height=1.25in,clip,keepaspectratio]{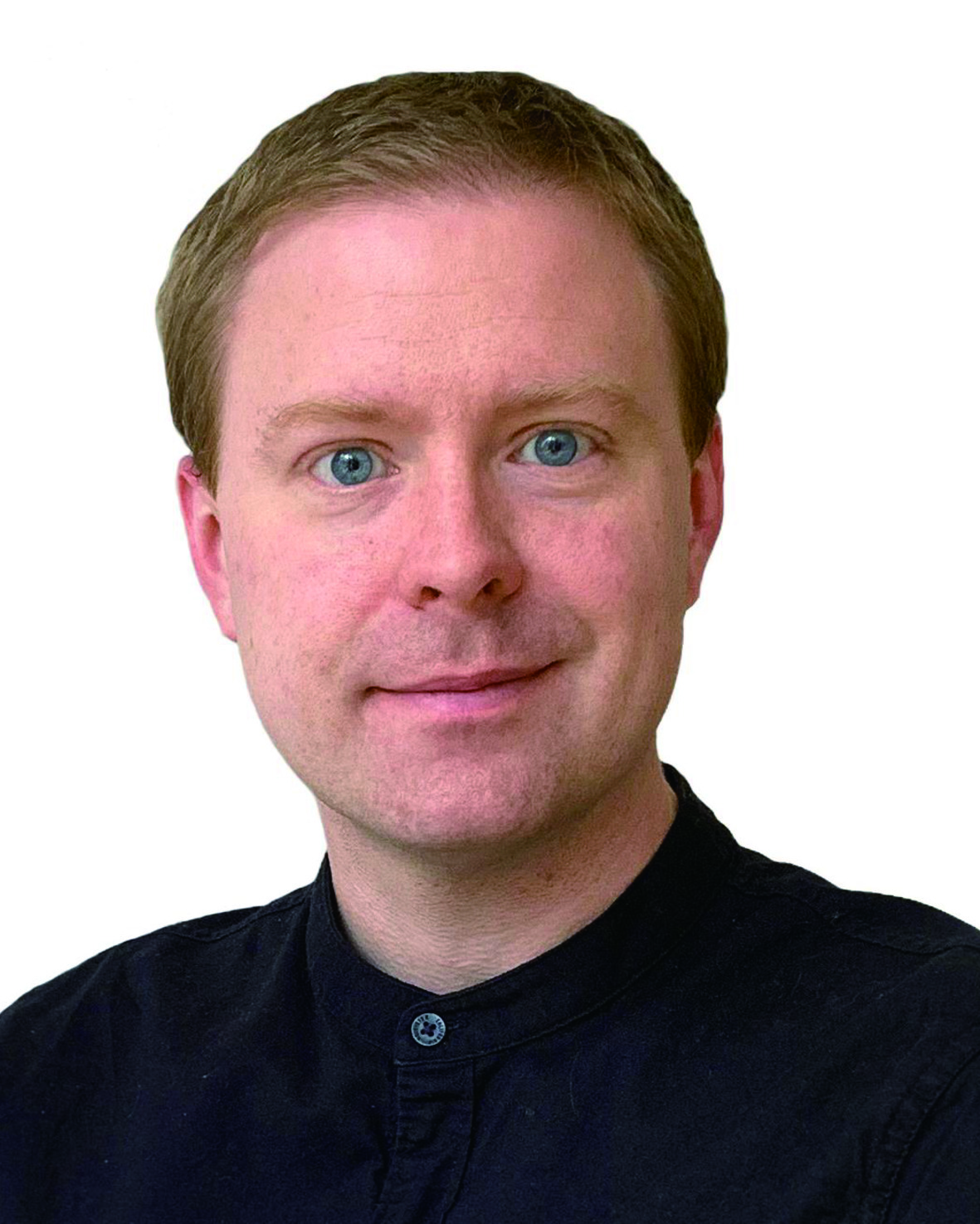}}]
    {Emil Bj{\"o}rnson} (S'07-M'12-SM'17-F'22) is a Full Professor of Wireless Communication at the KTH Royal Institute of Technology, Sweden. He received an M.S. degree in engineering mathematics from Lund University, Sweden, in 2007, and a Ph.D. degree in telecommunications from KTH in 2011. From 2012 to 2014, he was a post-doc at the Alcatel-Lucent Chair on Flexible Radio, SUPELEC, France. From 2014 to 2021, he held different professor positions at Link\"oping University, Sweden. He was a Visiting Full Professor at KTH in 2020-2021, before obtaining a tenured position in 2022.

    He has authored the textbooks \emph{Optimal Resource Allocation in Coordinated Multi-Cell Systems} (2013), \emph{Massive MIMO Networks: Spectral, Energy, and Hardware Efficiency} (2017), and \emph{Foundations of User-Centric Cell-Free Massive MIMO} (2021). He is dedicated to reproducible research and has made a large amount of simulation code publicly available. He performs research on MIMO communications, radio resource allocation, machine learning for communications, and energy efficiency. He is an Area Editor in IEEE Signal Processing Magazine.

    He has performed MIMO research for 16 years, his papers have received more than 25000 citations, and he has filed more than twenty patent applications. He is a host of the podcast Wireless Future and has a popular YouTube channel with the same name. He is an IEEE Fellow, a Wallenberg Academy Fellow, a Digital Futures Fellow, and an SSF Future Research Leader. He has received the 2014 Outstanding Young Researcher Award from IEEE ComSoc EMEA, the 2015 Ingvar Carlsson Award, the 2016 Best Ph.D. Award from EURASIP, the 2018 and 2022 IEEE Marconi Prize Paper Awards in Wireless Communications, the 2019 EURASIP Early Career Award, the 2019 IEEE Communications Society Fred W. Ellersick Prize, the 2019 IEEE Signal Processing Magazine Best Column Award, the 2020 Pierre-Simon Laplace Early Career Technical Achievement Award, the 2020 CTTC Early Achievement Award, the 2021 IEEE ComSoc RCC Early Achievement Award, and the 2023 IEEE ComSoc Outstanding Paper Award. He also co-authored papers that received Best Paper Awards at the conferences, including WCSP 2009, the IEEE CAMSAP 2011, the IEEE SAM 2014, the IEEE WCNC 2014, the IEEE ICC 2015, and WCSP 2017.
\end{IEEEbiography}

\end{document}